\documentclass[a4paper,UKenglish,cleveref, autoref, thm-restate]{lipics-v2021}

\pdfoutput=1 
\hideLIPIcs

\usepackage{rotate}
\usepackage{multirow}
\usepackage{multicol}
\usepackage{bussproofs}
\usepackage{amssymb}
\usepackage{mathtools}
\usepackage{graphics}
\usepackage{tikz-cd}
\usepackage{proof}
\usepackage{color}
\usepackage[title]{appendix}
\usepackage{cmll}  
\usetikzlibrary{positioning}
\usetikzlibrary{matrix}
\usetikzlibrary{calc}
\usetikzlibrary{backgrounds}
\usepackage{cmll} 
\usepackage{stmaryrd}
\usepackage{amsthm}
\usepackage{amsmath}
\usepackage{dsfont}
\usepackage{cite}

\newcommand{\bang}{\mathop{!}}
\newcommand{\quest}{\mathop{?}}
\newcommand\bb{\mathcal{B}}
\newcommand\bc{\mathcal{C}}
\newcommand\bd{\mathcal{D}}
\newcommand\be{\mathcal{E}}
\newcommand{\At}{\sf{At}}
\newcommand{\Rule}{\sf{Rule}}
\newcommand{\forms}{\sf{Form}}
\newcommand{\Fact}{\sf{Fact}}
\newcommand{\bs}{\mathbb{B}}
\newcommand{\ps}{\mathbb{P}}
\newcommand{\brl}{\llbracket}
\newcommand{\brr}{\rrbracket}
\newcommand{\myv}{\Vdash_\bb^\varnothing}

\newenvironment{bprooftree}
  {\leavevmode\hbox\bgroup}
  {\DisplayProof\egroup}

\makeatletter
\def\namedlabel#1#2{\begingroup
    #2%
    \def\@currentlabel{#2}%
    \phantomsection\label{#1}\endgroup
}
\makeatother

\title{From Phase Semantics to Base-extension Semantics (and back)}

\titlerunning{From Phase Semantics to Base-extension Semantics (and back)} 

\author{Ekaterina Piotrovskaya}{Department of Computer Science, University College London, UK}{kate.piotrovskaya.21@ucl.ac.uk}{https://orcid.org/0009-0009-4217-6948}{}

\authorrunning{E. Piotrovskaya} 

\Copyright{Ekaterina Piotrovskaya} 

\ccsdesc[500]{Theory of computation~Linear logic}
\ccsdesc[300]{Theory of computation~Proof theory}
\ccsdesc[300]{Theory of computation~Algebraic semantics}

\keywords{Linear logic, Phase semantics, Base-extension semantics, Equivalence of semantics} 

\nolinenumbers 

\begin{document}

\maketitle

\begin{abstract}
\noindent Linear logic admits a wide range of semantic presentations reflecting its resource-sensitive notion of consequence. One well-known example is phase semantics: an algebraic semantics in which formulas are interpreted in phase models, consisting of a phase space formed by a commutative monoid and a fixed subset, with respect to which an orthogonality relation is defined, and a valuation. A rather different and much more recent approach is given by base-extension semantics, which defines validity by inductively extending a provability relation on a base -- a set of inference rules over atomic propositions. We establish an equivalence between the two semantics by first defining bidirectional maps between bases and phase models, and then constructing an isomorphism between a phase model (resp. base) and its image under the composition of these maps. As a further contribution, we define the base-extension semantics clauses for the exponentials of linear logic. 

\end{abstract}

\section{Introduction}
Linear logic~\cite{DBLP:journals/tcs/Girard87} admits a wide range of semantic presentations, reflecting the fact that its notion of consequence is sensitive not only to formula formation but also to the structural behaviour of contexts. Unlike classical and intuitionistic systems, where semantic clauses can often be stated over sets of valuations or Kripke-style structures with relatively familiar closure conditions, the semantics of linear logic must account for the controlled use of assumptions.\footnote{Linear logic can also be seen as an abstract logic programming language~\cite{DBLP:conf/oopsla/AndreoliP90}, since it is sound and non-deterministically complete with respect to the logical interpretation of programs and has a proof-search strategy attached to it~\cite{DBLP:journals/logcom/Andreoli92,DBLP:journals/apal/MillerNPS91}. } This has led to several distinct semantic frameworks, each capturing different aspects of proofs and provability in the logic. These include, but are not limited to, coherent spaces~\cite{DBLP:journals/tcs/Girard87}, relational semantics~\cite{DBLP:journals/mscs/Falco03,DBLP:journals/tcs/Ehrhard12}, Kripke-style semantics~\cite{DBLP:journals/jsyml/AllweinD93,DBLP:journals/iandc/HodasM94}, categorical models~\cite{DBLP:journals/tcs/Abramsky93,DBLP:journals/mscs/Ehrhard93,DBLP:conf/fossacs/MaiettiPR00,Mellies09panorama,mellies:hal-00154229} and game semantics~\cite{DBLP:journals/jsyml/AbramskyJ94,DBLP:conf/lics/AbramskyM99}. 

Among the standard semantics for linear logic is \emph{phase semantics}~\cite{DBLP:journals/tcs/Girard87,DBLP:journals/iandc/FagesRS01,DBLP:conf/fsttcs/0001JS22}. In this semantics, formulas are interpreted over a \emph{phase space} consisting of a commutative monoid (whose elements are called \emph{phases}) and a fixed subset, with respect to which an orthogonality relation is defined. A pair of a phase space and a valuation function that induces the interpretation is called a \emph{phase model}. Each atom is interpreted as a \emph{fact} -- a set of phases that is equal to its closure induced by orthogonality. The interpretation is then inductively extended to formulas via the semantic clauses that mirror the resource-sensitive behaviour of the logical connectives. Phase semantics yields the standard soundness-and-completeness theorem for provability in linear logic via a canonical phase model built from contexts (i.e.~multisets of formulas)~\cite{DBLP:journals/apal/Girard93,okada1998introduction}. Phase semantics has, for instance, been used to prove decidability~\cite{dal2004phase, lafont1997finite,lafont1996undecidability} and cut admissibility results~\cite{okada1999phase,okada1998introduction}.

A rather different and much more recent approach is given by \emph{base-extension semantics}~\cite{DBLP:journals/synthese/Schroeder-Heister06,DBLP:journals/igpl/Sandqvist15,piecha2016completeness,DBLP:journals/jphil/PiechaSS15}. Base-extension semantics (B-eS) is a strand of \emph{proof-theoretic semantics}~\cite{sep-proof-theoretic-semantics,wansing2000idea} that provides an alternative approach to the meaning of logical operators. It is based on inferentialism -- a view according to which the meaning of expressions arises from their use in the system of inference and is defined by the rules of said system. The meaning is thus expressed in terms of proofs and provability: the concept of truth inherent to model-theoretic semantics is substituted with that of \emph{proofs}. 

In base-extension semantics, the starting point is a fixed atomic system -- a set of atomic propositions together with a collection of inference rules over them -- called a \emph{base}. The characterisation of consequence in B-eS is given via a judgement called \emph{support}, which is inductively defined on the structure of formulas, with the base case (i.e. support of atoms) determined by \emph{provability in a base}. In this sense, base-extension semantics replaces the usual starting point of model-theoretic semantics -- assignment of denotations or truth values -- with an inferential base, while retaining a compositional notion of semantic validity; hence, soundness and completeness resemble Tarskian semantics, but express correspondence between derivability in a logic and provability in bases rather than truth in models. For linear logic, however, this approach is not immediate. Because the logic is substructural, the shape of contexts cannot be treated as a secondary matter: multiplicity and structural behaviour of exponentials must already be reflected at the level of the semantic clauses, requiring B-eS to be capable of tracking the same structural distinctions as the proof system itself~\cite{DBLP:conf/tableaux/GheorghiuGP23,DBLP:journals/corr/abs-2504-08349,DBLP:journals/corr/abs-2402-01982}.

In this work, we establish an equivalence between these two semantics. At first sight, phase semantics and base-extension semantics are formulated in rather different terms. However, the two exhibit certain similarities: for instance, both are provability semantics, i.e., concerned with the provability of formulas, rather than the nature, shape, or distinct instances of proofs. Further, in both of them, the set of operational units of the semantics has a structure to it: algebraic in the case of phase semantics -- elements of the monoid in a phase space have a binary operation defined over them, and proof-theoretic in the case of B-eS -- atoms in a base are related by inference rules. Furthermore, both begin with fixing the interpretation of $\bot$ -- a chosen subset of the monoid for phase semantics, and a chosen atom for base-extension semantics. This makes it natural to ask whether they can be compared directly, and in particular whether one can establish a precise correspondence between them.

Such a comparison is useful for several reasons. First, an equivalence result shows that the two semantics not only define the same notion of validity, but are also intertranslatable on a structural level. This gives a correctness check on both frameworks: a semantics formulated by induction over proofs and a semantics formulated over phase spaces are able to express exactly the same inferences. Second, once a correspondence is established, results obtained in one setting can be transferred to the other. Completeness arguments, countermodel constructions, and structural properties that are transparent in phase semantics may then be reformulated in base-extension terms, while inductive arguments native to base-extension semantics may clarify the proof-theoretic content of phase model constructions. Lastly, it places base-extension semantics within the established semantic theory of linear logic, and thereby makes it possible to compare it systematically with other standard frameworks.

More generally, relating different semantics for the same logic is important because it separates essential structure from details of presentation. This provides a basis for viewing bases and phase models as categories and establishing the equivalence presented in this paper on a more abstract level, as well as connecting this work to~\cite{DBLP:journals/sLogica/PymRR25}. 

\medskip

\noindent
{\bf Contributions.} \ \ The main contribution of this paper is establishing the equivalence between phase semantics and base-extension semantics for classical linear logic. We observe that there is a systematic way of converting a phase model into a base and, similarly, constructing a phase model from a given base. We formalise this as bidirectional maps. We then introduce the notions of structure-preserving maps between bases -- \emph{base morphisms} -- and phase models -- \emph{phase model morphisms} -- and construct an isomorphism of phase models (resp. bases). As a contribution to the field of base-extension semantics, we extend the multiplicative-additive fragment of linear logic in~\cite{DBLP:journals/corr/abs-2504-08349} with the clauses for {exponentials $(!)$ and $(?)$.} 

\medskip

\noindent
{\bf Structure of the paper.} \ \ The paper is organised as follows. In Section~\ref{sec:back}, we give a brief overview of linear logic, phase semantics and base-extension semantics with relevant definitions and properties. In Section~\ref{sec:corr}, we define bidirectional maps between bases and phase models and construct an isomorphism between a phase model (resp. base) and its image under the composition of these maps. Finally, we conclude in Section~\ref{sec:conclusion}, discussing the results of this paper and proposing directions for future work. Two appendices complement the paper: proofs of soundness and completeness w.r.t. linear logic of the $(!)$ and $(?)$ base-extension semantics clauses are provided in Appendix~\ref{app:a}, and omitted proofs of lemmas in Section~\ref{sec:corr} are provided in Appendix~\ref{app:proofs}.

\medskip

\noindent
{\bf Notation.} \ \ Lowercase Latin letters ($p, q$) denote atoms; capital Latin letters ($L,K$) denote finite multisets of atoms; lowercase Greek letters ($\phi, \psi$) denote formulas; capital Greek letters ($\Gamma, \Delta$) denote multisets of formulas; and commas between multisets denote multiset union.

\section{Background} \label{sec:back}
\subsection{Linear Logic}
Classical linear logic~\cite{DBLP:journals/tcs/Girard87} (LL) is a resource-sensitive logic, meaning that formulas are consumed when used in proofs unless explicitly marked with the exponential modalities $\bang$ (bang) and $\quest$ (quest). Formulas marked with these exponentials behave \emph{classically}, i.e.~they can be contracted (duplicated) and weakened (erased) during proofs.

The propositional connectives of LL include the additive conjunction $\with$ and disjunction $\oplus$, as well as their multiplicative counterparts, tensor $\otimes$ and par $\parr$, along with their respective units. While in the intuitionistic setting implication is considered a separate connective, in the classical setting it can be expressed using $\parr$ and negation; we hence do not treat it as primitive. Note, however, that the results of this paper can be easily extended to implication, if one wishes to do so.

Below is the syntax of linear logic:

\par\nobreak\bgroup%
\begin{tikzpicture}[node distance=1ex]
  \node [matrix of math nodes] (gr) {
     \phi, \psi  & \Coloneq &
    \node(a){p}; & \mid & \node(imp){\ \ \phi^\bot}; & \mid & \node(tens){\phi \otimes  \psi}; & \mid & \mathsf{1} & \mid &
    \node(plus){\phi \oplus  \psi}; & \mid & \mathsf{0} &
    \mid & \node(bang){\mathsf{!} \phi}; \\
    & \node[right] {} ; & 
     \node(a'){ }; &  & 
     \node(imp'){ }; & \mid &
    \node(par){\phi \parr  \psi}; & \mid & \node(bot){\bot}; & \mid &
    \node(with){\phi \with  \psi}; & \mid & \node(top){\top}; &
     \mid & \node(qm){\mathsf{?} \phi}; \\
  } ;
  \node at ($(bang.north east)!.5!(qm.south east)+(1.8,0)$) {
    \refstepcounter{equation}
    \label{eq:gram}
  } ;
  \begin{scope}[on background layer]
    \fill[rounded corners,color=green!5!white]
       ($(a.north west) - (0.1,0)$) rectangle ($(a'.south east)+(.25,-.5)$) ;
    \node[align=flush left] at ($(a'.south west)!.5!(a'.south east)+(.15,-.2)$){
      \tiny\scshape \hspace{-1em} atoms
    } ;
      \fill[rounded corners,color=yellow!5!white]
      (imp.north west) rectangle ($(imp'.south east)+(.4,-.5)$) ;
    \node[align=flush left] at ($(imp'.south west)!.5!(imp'.south east)+(.15,-.2)$){
      \tiny\scshape \hspace{-0.85em} negation
    } ;
    \fill[rounded corners,color=blue!5!white]
       (tens.north west) rectangle ($(bot.south east)-(0,.5)$) ;
    \node at ($(par.south west)!.5!(bot.south east)-(0,.2)$) {
      \tiny\scshape multiplicatives
    } ;
    \fill [rounded corners,color=red!5!white]
       (plus.north west) rectangle ($(top.south east)-(0,.5)$) ;
    \node at ($(with.south west)!.5!(top.south east)-(0,.2)$) {
      \tiny\scshape additives
    } ;
    \fill [rounded corners,color=cyan!10!white]
       (bang.north west) rectangle ($(qm.south east)+(.2,-.5)$) ;
    \node[align=flush left] at ($(qm.south west)!.5!(qm.south east)+(.15,-.2)$) {
      \tiny\scshape exp.
    } ;
  \end{scope}
\end{tikzpicture}
\egroup\par\nobreak\noindent
\setcounter{equation}{0}

We denote the set of linear logic formulas over the set of atoms $\At$ by ${\forms}_{\At}$; the set of finite multisets of atoms by $\mathcal{M}(\At)$; and the set of finite multisets of formulas by $\mathcal{M}(\forms_{\At})$.

\subsection{Phase Semantics}
Phase semantics provides a provability-based interpretation of linear logic, where each formula is identified with a set of resource contexts (facts) that are sufficient to establish its proof. Hence, phase semantics is concerned with the provability of formulas, rather than the context or shape of the proofs (unlike, e.g.~coherent spaces -- a denotational semantics~\cite{DBLP:journals/tcs/Girard87}).%
\begin{definition}[Phase space]
    A \emph{phase space} is a pair $(M, \bot)$, where $(M,\cdot,1_M)$ is a commutative monoid and $\bot$ is a subset of $M$.
\end{definition}

For $X,Y \subseteq M$, define the following operations: ${X.Y \coloneqq \{x\cdot y \ | \ x \in X, y \in Y\}}$ and $X^\bot \coloneqq \{m \in M \ | \ \forall x \in X \ \ m \cdot x \in \bot \}$. For any subset $X\subseteq M$, it is the case that $X\subseteq X^{\bot\bot}$; and $X$ is called a \emph{fact} if $X^{\bot\bot} = X$. Note that $Y^\bot$ for any $Y \subseteq M$ is a fact. The set of facts is denoted by ${\Fact}_M$. By $J$ we denote a chosen submonoid of $M$ satisfying the weak idempotent property: $\{j\}^{\bot\bot} \subseteq \{j \cdot j \}^{\bot\bot}$ for all $j\in J$.

We define the following operations on facts, where $I \coloneqq 1 \cap J$:
\begin{multicols}{2}
\begin{description} \setlength\itemsep{0.5em}
    \item[{($\otimes$)}] $X \otimes Y \coloneqq (X.Y)^{\bot \bot}$
    \item[{($\parr$)}] $X \parr Y \coloneqq (X^{\bot}.Y^{\bot})^{\bot}$
    \item[{($1$)}] $1 \coloneqq \{1_M\}^{\bot\bot} (= \bot^\bot )$
    \item[{($\oplus$)}] $X \oplus Y \coloneqq (X \cup Y)^{\bot \bot}$
    \item[{($\with$)}] $X \with Y \coloneqq X \cap Y \qquad \qquad \qquad \qquad$ 
    \item[{($\top$)}] $\top \coloneqq M$
    \item[{($0$)}] $0 \coloneqq  \varnothing^{\bot \bot} (= \top^\bot)$
    \item[{($!$)}] $!X \coloneqq (X \cap I)^{\bot \bot}$
    \item[{($?$)}] $?X \coloneqq (X^\bot \cap I)^\bot$
\end{description}
\end{multicols}

\begin{example}
    Consider $(\mathbb{N}, +, 0)$ as an additive monoid over natural numbers, and fix $\bot = \{1\}$. For any $S \subseteq \mathbb{N}$, $S^\bot = \{ n \in \mathbb{N} \; | \; \forall s \in S \ \ n+s =1 \}$. Then $\{0\}^{\bot\bot} = \{1\}^\bot = \{0\}$ is a fact, but $\{3\}^{\bot\bot} = \varnothing^\bot = \mathbb{N} \neq \{3\}$ is not. Overall, ${\Fact}_\mathbb{N} = \{\varnothing,\{0\},\{1\}, \mathbb{N}\}$.
\end{example}

\begin{proposition} \label{prop:psproperties}
    Given a phase space $(M,\bot)$, the following hold for all $X, Y \in {\Fact}_M$:
    \begin{multicols}{2}
    \begin{enumerate}
        \item $X.X^\bot \subseteq \bot$
        \item $X \subseteq Y \Rightarrow Y^\bot \subseteq X^\bot$
        \item $X \subseteq Y^\bot \Leftrightarrow X.Y \subseteq \bot$
        \item $(X \cup Y)^\bot = X^\bot \cap Y^\bot$
    \end{enumerate}
    \end{multicols}
\end{proposition}

A phase model is a tuple $(M,\bot,v)$ of a phase space $(M,\bot)$ and a \emph{valuation} function $v: {\At} \rightarrow {\Fact}_M$. Linear logic formulas over $\At$ are interpreted in a phase model as follows.
\begin{definition}[Interpretation] \label{def:interpretation}
    Given a phase model $(M,\bot,v)$, the interpretation $(-)^v: {\forms}_{\At} \rightarrow {\Fact}_M$ is recursively defined as follows:
    \begin{align*}
        & a^v = v(a) \text{ for } a \in {\At} & & 0^v = \top^\bot \\
        & (A^\bot)^v = (A^v)^\bot & & (A \odot B)^v = A^v \odot B^v \text{ for } \odot \in \{ \otimes,\parr,\with,\oplus\} \\
        & \bot^v = \bot & & (\bang A)^v = \bang (A^v) \\
        & 1^v = \bot^\bot & & (\quest A)^v = \quest(A^v) \\
        & \top^v = M & & 
    \end{align*}
\end{definition}

\begin{definition}[Validity]
    Given a phase model $(M,\bot,v)$, a formula $A$ is \emph{valid} if $1_M \in A^v$ (equivalently, $1 \subseteq A^v$).
\end{definition}

Under the interpretation in Definition~\ref{def:interpretation}, linear logic is sound and complete with respect to phase models. We refer the reader to~\cite{DBLP:journals/tcs/Girard87,okada1998introduction} for the details.

\begin{definition}[Definable facts]\label{def:definablefact}
    Given a phase model $(M,\bot,v)$, a fact $X \in {\Fact}_M$ is \emph{definable} if $X \in im((-)^v)$; equivalently, $X = A^v$ for some $A \in \forms_{\At}$.
\end{definition}
We denote the set of all definable facts of a phase model $(M,\bot,v)$ by $D_M$.

The primary objects of interest are facts definable in a given phase model: in Definition~\ref{def:interpretation}, we interpret all linear logic formulas over such facts. We, therefore, denote a phase model by $P = (D_M, I, \bot, v)$, leaving the monoid structure and the choice of $J$ implicit.

\subsection{Base-extension Semantics}
Base-extension semantics (B-eS) is founded on an inductively defined judgment called {\em support}, which mirrors the syntactic structure of formulas. The inductive definition begins with a base case: the support of atomic propositions is determined by derivability in a given {\em base} -- a specified collection of inference rules that govern atomic propositions. Sandqvist~\cite{DBLP:journals/igpl/Sandqvist15} introduced a sound and complete formulation of B-eS for intuitionistic propositional logic.

Ideally, one wishes to represent rules in a tree-based, sequent-style format that aligns naturally with the structure of multiplicative and additive rules of LL. However, such a presentation hides the difference between discharged and undischarged hypotheses, which is crucial in natural deduction. We opt for the best of both worlds and use the sequent presentation, while keeping hypotheses being discharged explicit via splitting the premises of each sequent into two zones: $A;B$. Formulas in $A$ denote context, while those in $B$ are discharged during an application of a given rule. Note that this notational convention differs from its usual meaning in the linear logic setting, where such a split of contexts normally refers to separating (multisets of) linear formulas from the (sets of) classical or intuitionistic ones (see, e.g.,~\cite{DBLP:journals/apal/Girard93}). We proceed to make this formal.

Fix a set $\At$. A \emph{sequent} $L;K \vdash p$ over $\At$ is a triple consisting of the multisets $L$ and $K$ with elements from $\At$ and an element $p \in \At$. If $K = \varnothing$, we write $L \vdash p$ for $L;\varnothing \vdash p$. A \emph{rule} over $\At$ is a finite non-empty list of sequents $\{L_i;K_i \vdash p_i\}_{i=1,\dots,n} \cup \{L_{n+1}\vdash p_{n+1}\}$ such that $L_{n+1} \subseteq \biguplus_{i=1}^n L_i$. We write a rule as
\begin{center}
    \begin{bprooftree}
    \AxiomC{$L_1;K_1 \vdash p_{1}$}
    \AxiomC{$\ldots$}
    \AxiomC{$L_n;K_n \vdash p_{n}$}
    \TrinaryInfC{$L_{n+1} \vdash p_{n+1}$}
    \end{bprooftree}
\end{center}

\begin{definition}[Base] \label{def:base}
A {\em base} $\mathcal{B}$ is a triple $(\At_\bb, \Rule_\bb, \flat)$, where $\At_\bb$ is a set whose elements are called atoms, $\Rule_\bb$ is a set of rules over $\At_\bb$ and $\flat: {\At} \rightarrow {\At}_\bb$ is a valuation function. 

A base is closed under rules of the following shape for all $p,r \in \At_\bb$, $K,L \in \mathcal{M}(\At_\bb)$:

\begin{center}
    \begin{bprooftree}
    \def\ScoreOverhang{0.5pt}
        \AxiomC{}
        \noLine
        \UnaryInfC{}
        \RightLabel{\small{Ax}}
        \UnaryInfC{$p \vdash p$}
    \end{bprooftree}
    \quad
    \begin{bprooftree}
    \def\ScoreOverhang{0.5pt}
        \AxiomC{$K \vdash p$}
        \AxiomC{$L; p \vdash r$}
        \RightLabel{\small{Subs}}
        \BinaryInfC{$K,L \vdash r$}
    \end{bprooftree}
\end{center}

Finally, $\bot \in \At_\bb$ for every base $\bb$.
\end{definition}

Linear logic formulas over $\At$ are interpreted in a base as follows.
\begin{definition}[Interpretation] \label{def:baseinterpretation}
    Given a base $({\At}_\bb, {\Rule}_\bb,\flat)$, the interpretation $(-)^\flat: {\forms}_{\At} \rightarrow {\forms}_{\At_\bb}$ is recursively defined as follows:
    \begin{align*}
        & a^\flat = \flat(a) \text{ for } a \in {\At} & & (A \odot B)^\flat = A^\flat \odot B^\flat \text{ for } \odot \in \{ \otimes,\parr,\with,\oplus\} \\
        & u^\flat = u \text{ for } u \in \{ \bot,1,\top,0\} & & (\bang A)^\flat = \bang (A^\flat) \\
        & (A^\bot)^\flat = (A^\flat)^\bot & & (\quest A)^\flat = \quest(A^\flat)
    \end{align*}
\end{definition}

\begin{remark}
    The notion of a valuation does not normally appear in B-eS literature, and interpretation is left implicit. However, one may notice that the usual type of bases considered is those over $\At$, in which case the valuation is simply the identity.
\end{remark}

In order to fully capture the behaviour of LL in our semantics, we aspire to distinguish between atoms that can behave structurally (i.e.~be weakened and contracted) and those that cannot, mimicking the distinction between exponential and non-exponential formulas of LL. As one would expect, an atom behaves structurally if the rules of both weakening and contraction for it are present in a given base. If an atom happens to satisfy these conditions, we call it \emph{structural}.
\begin{definition}[Structural atoms] \label{def:structuralatoms}
    Given a base $\bb$, we say that $p \in \At_\bb$ is \emph{structural in $\bb$} if the following rules are present in $\bb$ for all $L,K$:
    \begin{center}
        \begin{bprooftree}
        \def\ScoreOverhang{0.5pt}
            \AxiomC{$L \vdash p$}
            \AxiomC{$K \vdash \bot$}
            \RightLabel{\small Wk}
            \BinaryInfC{$L,K \vdash \bot$}
        \end{bprooftree}
        \quad
        \begin{bprooftree}
        \def\ScoreOverhang{0.5pt}
            \AxiomC{$L \vdash p$}
            \AxiomC{$K;p,p \vdash \bot$}
            \RightLabel{\small Ctr}
            \BinaryInfC{$L,K \vdash \bot$}
        \end{bprooftree}
    \end{center}
    A multiset $S$ is \emph{structural in $\bb$} if each $s \in S$ is.
\end{definition}

Note that we do not introduce a separate type of atoms in $\At_\bb$ -- any atom may (or may not) become structural subject to $\Rule_\bb$.

\begin{definition}[Extensions]
A base $\mathcal{C}$ is an {\em extension} of a base $\mathcal{B}$ (written $\mathcal{C}\supseteq\mathcal{B}$), if $\At_\bc = \At_\bb$ and $\Rule_\bc \supseteq \Rule_\bb$ and $\flat_\bc = \flat_\bb$.
\end{definition}

\begin{definition}[Derivability in a base] Given a base $\mathcal{B}$, the \emph{atomic derivability} relation $\vdash_{\mathcal{B}} \; \subseteq \mathcal{M}({\At}_\mathcal{B}) \times {\At}_\mathcal{B}$ is defined as follows:
\begin{itemize}
    \item $p \vdash_\mathcal{B} p$ for all $p \in \At_\bb$;
    \item Given the following rule in $\mathcal{B}$:
\begin{center}
    \begin{prooftree}
    \AxiomC{$L_1;K_1 \vdash p_{1}$}
    \AxiomC{$\ldots$}
    \AxiomC{$L_n;K_n \vdash p_{n}$}
    \TrinaryInfC{$L_{n+1} \vdash p_{n+1}$}
    \end{prooftree}
\end{center}
if $L_{i},K_i \vdash_{\mathcal{B}} p_{i}$ hold for all $1 \leq i \leq n$, then $L_{n+1} \vdash_{\mathcal{B}} p_{n+1}$ holds.
\end{itemize}
\end{definition}

\begin{example}
    Let $p=$ {\em ``There is electricity outage''}, $q=$ {\em ``Wi-fi is not working''}, $r=$ {\em ``Eduroam is down''}, $s=$ {\em ``I will miss the CSL deadline''} and $\bb = (\{p,q,r,s\}, {\Rule}_\bb, \flat)$ be the base, whose ${\Rule}_\bb$ consists of the rules below: 
    \begin{center}
        \begin{bprooftree}
        \def\ScoreOverhang{0.5pt}
            \AxiomC{}
            \UnaryInfC{$p \vdash q$}
        \end{bprooftree}
        \quad
        \begin{bprooftree}
        \def\ScoreOverhang{0.5pt}
            \AxiomC{}
            \UnaryInfC{$q \vdash r$}
        \end{bprooftree}
        \quad
        \begin{bprooftree}
        \def\ScoreOverhang{0.5pt}
            \AxiomC{$p \vdash r$}
            \UnaryInfC{$\vdash s$}
        \end{bprooftree}
    \end{center}
    Then the following is a deduction in $\bb$ concluding $\vdash_\bb s$:
    \begin{prooftree}
    \AxiomC{}
    \UnaryInfC{$p \vdash q$}
        \AxiomC{}
    \UnaryInfC{$q \vdash r$}
         \RightLabel{\scriptsize{Subs}}
    \BinaryInfC{$p \vdash r$}
    \UnaryInfC{$\vdash s$}
\end{prooftree}
\end{example}

We next define the support relation, which is reducible to derivability in a given base and its extensions.
\begin{definition}[Support] \label{def:support}
    Given a base $\bb$, the \emph{support relation} $\Vdash^{L}_{\mathcal{B}} \; \subseteq \mathcal{M}({\forms}_{{\At}_\mathcal{B}}) \times {\forms}_{{\At}_\mathcal{B}}$ is defined as follows, where all multisets of atoms and of formulas are assumed to be finite:
    
    \begin{description} \setlength\itemsep{0.7em}
        \item[\namedlabel{eq:supp-at}{(At)}] $\Vdash^{L}_{\mathcal{B}} p$ $\; \Leftrightarrow \;$ for all $\mathcal{C} \supseteq \mathcal{B}$ and $K$, if $p, K \vdash_{\mathcal{C}} \bot$ then $L, K \vdash_{\mathcal{C}} \bot$, for atomic $p$;
     
        \item[\namedlabel{eq:supp-tensor}{$(\otimes)$}] $\Vdash^{L}_{\mathcal{B}} \phi \otimes \psi$ $\; \Leftrightarrow \;$ for all $\mathcal{C} \supseteq \mathcal{B}$ and $K$, if $\phi,\psi \Vdash^{K}_{\mathcal{C}} \bot$ then $\Vdash^{L,K}_{\mathcal{C}} \bot$;
    
        \item[\namedlabel{eq:supp-imply}{$(^\bot)$}] $\Vdash^{L}_{\mathcal{B}} \phi^\bot $ $\; \Leftrightarrow \;$ $\phi \Vdash^{L}_{\mathcal{B}} \bot$;
    
        \item[\namedlabel{eq:supp-1}{$(1)$}] $\Vdash^{L}_{\mathcal{B}} 1$ $\; \Leftrightarrow \;$ for all $\mathcal{C} \supseteq \mathcal{B}$ and $K$, if $\Vdash^{K}_{\mathcal{C}} \bot$ then $\Vdash^{L, K}_{\mathcal{C}} \bot$; 
    
        \item[\namedlabel{eq:supp-parr}{$(\parr)$}] $\Vdash^{L}_{\mathcal{B}} \phi \parr \psi$ $\; \Leftrightarrow \;$ for all $\mathcal{C} \supseteq \mathcal{B}$ and $K, M$, if $\phi \Vdash^{K}_{\mathcal{C}} \bot$ and $\psi \Vdash^{M}_{\mathcal{C}} \bot$ then $\Vdash^{L, K, M}_{\mathcal{C}} \bot$;

        \item[\namedlabel{eq:supp-and}{$(\with)$}] $\Vdash^{L}_{\mathcal{B}} \phi \with \psi$ $\; \Leftrightarrow \;$ $\Vdash^{L}_{\mathcal{B}} \phi$ and $\Vdash^{L}_{\mathcal{B}} \psi$;
        
        \item[\namedlabel{eq:supp-plus}{$(\oplus)$}] $\Vdash^{L}_{\mathcal{B}} \phi \oplus \psi$ $\; \Leftrightarrow \;$ for all $\mathcal{C} \supseteq \mathcal{B}$ and $K$, if $\phi \Vdash^{K}_{\mathcal{C}} \bot$ and $\psi \Vdash^{K}_{\mathcal{C}} \bot$ then $\Vdash^{L, K}_{\mathcal{C}} \bot$;

        \item[\namedlabel{eq:supp-top}{$(\top)$}] $\Vdash^{L}_{\mathcal{B}} \top$ for all $\mathcal{B}$ and $L$;

        \item[\namedlabel{eq:supp-0}{$(0)$}] $\Vdash^{L}_{\mathcal{B}} 0$ $\; \Leftrightarrow \;$ $\Vdash^{L, K}_{\mathcal{B}} \bot$ for all $K$;

        \item[\namedlabel{eq:supp-bang}{$(\bang)$}] $\Vdash^{L}_{\mathcal{B}} \bang\phi$ $\; \Leftrightarrow \;$ for all $\mathcal{C} \supseteq \mathcal{B}, \mathcal{D} \supseteq \mathcal{C}$, all $K$ and all $S$ structural in $\bb$, if $(\Vdash^{S}_{\mathcal{D}} \phi \Rightarrow \ \Vdash^{S,K}_{\mathcal{D}} \bot)$ then $\Vdash^{L,K}_{\mathcal{C}} \bot$;
            
        \item[\namedlabel{eq:supp-quest}{$(\quest)$}] $\Vdash^{L}_{\mathcal{B}} \quest\phi$ $\; \Leftrightarrow \;$ for all $\mathcal{C} \supseteq \mathcal{B}$ and all $S$ structural in $\bb$, if $\phi \Vdash^{S}_{\mathcal{C}} \bot$ then $\Vdash^{S,L}_{\mathcal{C}} \bot$;

        \item[\namedlabel{eq:supp-inf}{(Inf)}] $\Gamma \Vdash^{L}_{\mathcal{B}} \phi$ $\; \Leftrightarrow \;$ for all $\mathcal{C} \supseteq \mathcal{B}$ and all $K_i$, if $\Gamma = \{\psi_{1}, \dots , \psi_{n}\}$ and $\Vdash^{K_i}_{\mathcal{C}} \psi_{i}$ for $1 \leq i \leq n$, then $\Vdash^{L, K_1, \dots, K_n}_{\mathcal{C}} \phi$. 
    \end{description}
\end{definition}

The clauses~\ref{eq:supp-bang} and~\ref{eq:supp-quest} for classical linear logic have hitherto not appeared in the literature and are introduced here for the first time. We compare these to the existing $(\bang)$ clause for intuitionistic linear logic~\cite{DBLP:journals/corr/abs-2402-01982} in Section~\ref{sec:conclusion}.

\begin{definition}[Validity in a base] \label{def:validity}
    Given a base $\bb$, we say that an inference from $\Gamma$ to $\phi$ is \emph{valid in $\bb$} if $\Gamma \Vdash_\bb^\varnothing \phi$.
\end{definition}

\begin{definition}[Validity] 
    An inference from $\Gamma$ to $\phi$ is \emph{valid} if $\Gamma \Vdash^{\varnothing}_{\mathcal{B}} \phi$ for all $\mathcal{B}$.
\end{definition}

We read $\Gamma \Vdash^{K}_{\mathcal{B}} \phi$ as ``the base $\mathcal{B}$ supports an inference from $\Gamma$ to $\phi$ relative to the multiset~$K$''. Note that $\Gamma \Vdash^{K}_{\mathcal{B}} \phi$ is equivalent to $\Gamma, K \Vdash^{\varnothing}_{\mathcal{B}} \phi$~\cite{DBLP:journals/corr/abs-2504-08349}. When $\phi \Vdash_\bb^\varnothing \psi$ and $\psi \Vdash_\bb^\varnothing \phi$, we say that $\phi$ and $\psi$ are \emph{intersupported} in $\bb$ and write $\phi \dashv \vdash_\bb \psi$.

Finally, we show how structural atoms allow us to mimic the behaviour of the introduction rule for $(\bang)$.

\begin{lemma} \label{lemma:phitobangphi}
    Let $S$ be structural in $\bb$. Then, $\Vdash_\bb^S \phi$ implies $\Vdash_\bb^S \bang \phi$.
\end{lemma}

\begin{proof}
    Assume, for an arbitrary $\bc \supseteq \bb$, all $\bd \supseteq \bc$, all structural $T$ and an arbitrary $M$, that $\Vdash_\bd^T \phi$ implies $\Vdash_\bd^{T,M} \bot$. Since $\Vdash_\bb^S \phi$, hence by monotonicity, $\Vdash_\bc^S \phi$, $S$ structural and $\bc$ its own extension, conclude $\Vdash_\bc^{S,M} \bot$. Since $\bc \supseteq \bb$ was arbitrary and $\bd \supseteq \bc$ was such that $\Vdash_\bd^T \phi$ implies $\Vdash_\bd^{T,M} \bot$, obtain $\Vdash_\bb^S \bang \phi$ by~\ref{eq:supp-bang}, as required.
\end{proof}

Linear logic is sound and complete with respect to B-eS. We refer the reader to~\cite{DBLP:journals/corr/abs-2504-08349} (multiplicative-additive fragment) and to Appendix~\ref{app:a} (exponentials) for details.

\section{Correspondence} \label{sec:corr}
In this section, we establish the equivalence of the two semantics. First, we observe that there is a systematic way of converting a phase model into a base and, similarly, constructing a phase model from a given base. Formally, this amounts to defining the bidirectional maps between a set of bases and a set of phase models.  Then, we introduce two types of structure-preserving maps: \emph{base morphisms} and \emph{phase model morphisms}. Finally, we construct an isomorphism between a phase model and its image under the composition of the bidirectional maps. We proceed to do the same in the base setting, where the resulting isomorphism is up to \emph{intersupport} -- as expected from a system founded on proof theory.

We write $\bs$ for the set of bases and $\ps$ for the set of phase models.
\subsection{From Phase semantics to Base-extension semantics}
As the very first step, given a phase model, we wish to construct a base out of it. To do so, according to Definition~\ref{def:base}, we have to specify a set of atoms, a set of rules, and a valuation.
\begin{definition}[Phase to Base] \label{def:phasetobase}
    We define the map $(-)_\mathbb{B}: \ps \rightarrow \bs$ as $(D_M, I, \bot, v) = P \mapsto P_\mathbb{B} = (\At_\mathbb{B},\Rule_\mathbb{B},\flat)$ by letting $\flat \coloneqq v$, ${\At}_\mathbb{B} \coloneqq D_M$ and $\Rule_\mathbb{B}$ consist of the following rules (for each atomic multiset $L,L_i,K$ and $S$), where $a,a_i,b_j,b \in D_M$ and $n,m \geq 1$:
    \begin{align*}
        & & & \begin{bprooftree}
        \def\ScoreOverhang{0.5pt}
            \AxiomC{$L; a^\bot \vdash \bot$}  
            \UnaryInfC{$L \vdash a$}
        \end{bprooftree}, \\
        \text{if } & \bigotimes_{i=1}^n a_i \subseteq \bigotimes_{j=1}^m b_j \text{ then} & & \begin{bprooftree}
        \def\ScoreOverhang{0.5pt}
            \AxiomC{$L_1 \vdash {a_1}$} 
            \AxiomC{$\dots$} 
            \AxiomC{$L_n \vdash {a_n}$} 
            \AxiomC{$K;b_1,\dots,b_m \vdash \bot$} 
            \QuaternaryInfC{$L_1,\dots,L_n,K \vdash \bot$}
        \end{bprooftree}, \\
        \text{if } & \bigcap_{i=1}^n a_i \subseteq b \text{ then} & & \begin{bprooftree}
        \def\ScoreOverhang{0.5pt}
            \AxiomC{$L \vdash {a_1}$} 
            \AxiomC{$\dots$} 
            \AxiomC{$L \vdash {a_n}$} 
            \AxiomC{$K;b \vdash \bot$} 
            \QuaternaryInfC{$L,K \vdash \bot$}
        \end{bprooftree}, \\
        \text{if } & b = \bang a \text{ then} & & \begin{bprooftree}
        \def\ScoreOverhang{0.5pt}
            \AxiomC{$L \vdash b$}
            \AxiomC{$K \vdash \bot$}
            \BinaryInfC{$L, K \vdash \bot$}
        \end{bprooftree} \; \text{ and } \begin{bprooftree}
        \def\ScoreOverhang{0.5pt}
            \AxiomC{$S \vdash a$}  
            \UnaryInfC{$S \vdash b$}
        \end{bprooftree}.
    \end{align*}
\end{definition}
We refer to the rule of the first shape as \textit{reductio ad absurdum}.

\begin{remark}
    Note that if $J$ is maximal in a given phase model $P$, the definition above can be simplified by removing the weakening rule in the $(\bang)$ case and letting $n,m \geq 0$.
\end{remark}

From now on, we will slightly abuse notation and write ${\forms}_\bb$ for ${\forms}_{\At_\bb}$ to avoid stacking subscripts for, e.g., $\bb = P_\mathbb{B}$. We proceed to give two examples of the effect of the map from Definition~\ref{def:phasetobase} on phase models. We consider two phase spaces: one, whose set of definable facts is the same across all corresponding models, and one which gives rise to phase models with distinct sets of definable facts. 

\begin{example}
    Consider $(\mathbb{N}, +, 0)$ as an additive monoid over natural numbers, and fix $\bot = \{0,1\}$. It forms a phase space $(\mathbb{N}, \bot)$ whose set of facts is ${\Fact}_\mathbb{N} = \{\varnothing,\{0\},\{0,1\}, \mathbb{N}\}$. Now, for any valuation $v$ we have that $\bot^v = \{0,1\}, 1^v = \{0\}, \top^v = \mathbb{N}$ and $0^v = \varnothing$. Hence, $D_\mathbb{N} = {\Fact_\mathbb{N}}$ for any choice of $v$. For a fixed $v$, denote $(D_\mathbb{N},I,\bot,v) \eqqcolon P$; then $P_\mathbb{B}$ is a base whose set of atoms is $D_\mathbb{N}$ and valuation is $v$. Below are some examples of the rules in $P_\mathbb{B}$ (where $\bot = \{0,1\}$) and the support statements they result in (shown underneath):

    \begin{center}
        \begin{bprooftree}\small
    \def\ScoreOverhang{0.5pt}
        \AxiomC{$L \vdash \varnothing$} 
        \AxiomC{$K;\{0\} \vdash \bot$} 
        \BinaryInfC{$L,K \vdash \bot$}
    \end{bprooftree}   
    \
    \begin{bprooftree}\small
    \def\ScoreOverhang{0.5pt}
        \AxiomC{$L_1 \vdash \bot$}  
        \AxiomC{$L_2 \vdash \mathbb{N}$} 
        \AxiomC{$K;\mathbb{N} \vdash \bot$} 
        \TrinaryInfC{$L_1,L_2,K \vdash \bot$}
    \end{bprooftree}
    \
    \begin{bprooftree}\small
    \def\ScoreOverhang{0.5pt}
        \AxiomC{$L \vdash \{0\}$}  
        \AxiomC{$L \vdash \bot$} 
        \AxiomC{$K;\{0\} \vdash \bot$} 
        \TrinaryInfC{$L,K \vdash \bot$}
    \end{bprooftree}
    \end{center}

    \begin{center}
        $\varnothing \Vdash_{P_\mathbb{B}} \{0\}$
        \qquad \qquad \quad \quad
        $\bot \otimes \mathbb{N} \Vdash_{P_\mathbb{B}} \mathbb{N}$
        \qquad \qquad \qquad \qquad 
        $\{0\} \with \bot \Vdash_{P_\mathbb{B}} \{0\}$
    \end{center}
\end{example}

\begin{example}
    Consider a commutative monoid over the power set $\mathcal{P}(\{a,b\})$, whose operation is given by intersection, i.e.~$(\mathcal{P}(\{a,b\}), \cap, \{a,b\})$. Fix $\bot = \{\varnothing\}$. It forms a phase space $P = (\mathcal{P}(\{a,b\}), \bot)$ whose set of facts is ${\Fact}_P = \{\{\varnothing\},\{\varnothing, \{a\}\},\{\varnothing, \{b\}\}, \mathcal{P}(\{a,b\})\}$. 
    
    Now, let $v_1$ be defined as $a \mapsto \mathcal{P}(\{a,b\})$ for all $a$. Denote the resulting phase model $(D_{P_1},I,\bot, v_1) \eqqcolon P_1$. We have that $D_{P_1} = \{\{\varnothing\},\mathcal{P}(\{a,b\})\}$: notice that $\{\varnothing\} = \mathcal{P}(\{a,b\})^\bot$, but facts $\{\varnothing, \{a\}\}$ and $\{\varnothing, \{b\}\}$ cannot be obtained via the interpretation induced by $v_1$, hence are not definable. Then $(P_1)_\mathbb{B}$ is a base whose set of atoms is $D_{P_1}$. We give an example of a rule in $(P_1)_\mathbb{B}$ and the resulting support statement (recall that $\bot = \{\varnothing\}$):
    \begin{center}
    \begin{bprooftree}
    \def\ScoreOverhang{0.5pt}
        \AxiomC{$L; \bot \vdash \bot$}  
        \UnaryInfC{$L \vdash \mathcal{P}(\{a,b\})$}
    \end{bprooftree} \qquad \qquad \qquad
    $\bot^\bot \Vdash_{(P_1)_\mathbb{B}} \mathcal{P}(\{a,b\})$
    \end{center}

    On the other hand, if we consider a valuation $v_2$ given by $a \mapsto \{\varnothing, \{a\}\}$ for all $a$, with the corresponding phase model $(D_{P_2},I,\bot, v_1) \eqqcolon P_1$, we have that $D_{P_2} = {\Fact}_P$: notice that $\{\varnothing, \{b\}\} = \{\varnothing, \{a\}\}^\bot$ and $\bot = \{\varnothing\}$ and $\mathcal{P}(\{a,b\}) = \{\varnothing\}^\bot$. Then $(P_2)_\mathbb{B}$ is a base whose set of atoms is $D_{P_2}$. We give an example of a rule in $(P_2)_\mathbb{B}$ and the resulting support statement:
    \begin{center}
    \begin{bprooftree}
    \def\ScoreOverhang{0.5pt}
        \AxiomC{$L \vdash  \{\varnothing, \{a\}\}$} 
        \AxiomC{$K; \{\varnothing, \{a\}\},  \{\varnothing, \{a\}\} \vdash \bot$} 
        \BinaryInfC{$L,K \vdash \bot$}
    \end{bprooftree} \qquad 
    $\{\varnothing, \{a\}\} \Vdash_{(P_2)_\mathbb{B}} \{\varnothing, \{a\}\} \otimes \{\varnothing, \{a\}\}$ 
    \end{center}
    \vspace{0.5em}
\end{example}

Given $P \in \mathbb{P}$, formulas over the induced base $P_\mathbb{B}$ can be interpreted in the original phase model $P$: the function $\brl-\brr: {{\forms}_{P_\mathbb{B}}} \rightarrow P$ is defined on atoms (i.e.~definable facts of $P$) by $m \mapsto m$ and inductively extended to all formulas by interpreting the connectives in the phase space in the usual way (cf.~Definition~\ref{def:interpretation}). By slight abuse of notation, we also define $\brl-\brr$ on multisets of atoms by $\brl L \brr \coloneqq \bigotimes_{l \in L} \brl l \brr$, where the tensor is in $P$. As a special case, we obtain $\brl \varnothing \brr = 1$.

The benefit of structural atoms lies in their exponential-like behaviour, which plays a crucial role in defining the semantics for modalities of LL. This behaviour becomes even more apparent in a base obtained from a phase space in the lemma below.

\begin{lemma}\label{lmm:structatinb}
    An atom $b \in P_\mathbb{B}$ is structural if and only if $b = \brl \bang \alpha \brr$ for some $\alpha \in {\forms}_{P_\mathbb{B}}$.
\end{lemma}

\begin{proof}
    In order for an atom $b$ to be structural, the rules of weakening and contraction have to be present for it in $P_\mathbb{B}$ (see Definition~\ref{def:structuralatoms}). Notice that for any $b = \brl \bang \alpha \brr$ the rule on the left is present in $P_\mathbb{B}$. Further, since $\brl \bang \alpha \brr \subseteq \brl \bang \alpha \brr \otimes_{P_\mathbb{B}} \brl \bang \alpha \brr$, the rule on the right is present in $P_\mathbb{B}$:
    \begin{center}
        \begin{bprooftree}
        \def\ScoreOverhang{0.5pt}
            \AxiomC{$L \vdash \brl \bang \alpha \brr$}
            \AxiomC{$K \vdash \bot$}
            \BinaryInfC{$L, K \vdash \bot$}
        \end{bprooftree}
        \quad
        \begin{bprooftree}
        \def\ScoreOverhang{0.5pt}
            \AxiomC{$L \vdash \brl \bang \alpha \brr$}
            \AxiomC{$K;\brl \bang \alpha \brr,\brl \bang \alpha \brr \vdash \bot$}
            \BinaryInfC{$L,K \vdash \bot$}
        \end{bprooftree}
    \end{center}
    Thus, $\brl \bang \alpha \brr$ are structural in $P_\mathbb{B}$ for any $\alpha$. Next, we show that atoms of such shape are the only structural atoms in $P_\mathbb{B}$.
    
    By the definition of $P_\mathbb{B}$, the rule of weakening is introduced for $b$ only if $b = \bang_{P_\mathbb{B}} a$ for some $a \in D_M$ (this is the only type of rule in the base with two premises and no hypothesis discharge), which is equivalent to $b = \brl \bang \alpha \brr$ for some $\alpha \in {\forms}_{P_\mathbb{B}}$ by the definition of $\brl - \brr$. By the definition of $P_\mathbb{B}$, the rule of contraction is introduced for an atom $b$ if $b \subseteq b \otimes_{P_\mathbb{B}} b$, which does hold for $b = \brl \bang \alpha \brr$. Thus, the only atoms for whom both structural rules exist in the base must be of shape $\brl \bang \alpha \brr$ for some $\alpha \in {\forms}_{P_\mathbb{B}}$.
\end{proof}

In the definition of the support relation, the atomic derivability need not coincide with the atomic support. In the lemma below, we demonstrate that in a base obtained from a phase space, it does. 

\begin{lemma} \label{lmm:deriveqsupport}
    Given $P \in \mathbb{P}$, $\Vdash^L_{P_\mathbb{B}} p$ if and only if $L \vdash_{P_\mathbb{B}} p$ for atomic $p$.
\end{lemma}

\begin{proof}
    Denote $\bb \coloneqq P_\mathbb{B}$. The right-to-left direction is immediate. Now, assume that $\Vdash^L_\bb p$. Since $\brl p \brr \otimes_M \brl p^\bot \brr \subseteq \brl \bot \brr$, the following is a rule in $\bb$:
    \begin{prooftree}
    \def\ScoreOverhang{0.5pt}
    \AxiomC{$p \vdash p$}
    \AxiomC{$p^\bot \vdash p^\bot$}
    \AxiomC{$\varnothing;\bot \vdash \bot$}
    \TrinaryInfC{$p, p^\bot \vdash \bot$}
\end{prooftree}
All premises are axioms, hence $p, p^\bot \vdash_\bb \bot$. Since also $\Vdash^L_\bb p$, by~\ref{eq:supp-at}, we conclude $L, p^\bot \vdash_\bb \bot$. By applying \textit{reductio ad absurdum}, conclude $L \vdash_\bb p$.
\end{proof}

Definition~\ref{def:phasetobase} guarantees that inclusion of facts in a phase space persists to the base -- if $a \subseteq b$, then $b$ is derivable from $a$ in the resulting base. We show that the converse holds too -- derivability in a base cannot introduce inferences that have not already been present in the phase space. 

\begin{lemma} \label{lmm:inclusion}
    Given $P \in \mathbb{P}$, if $L \vdash_{P_\mathbb{B}} p$ then $\llbracket L \rrbracket \subseteq \brl p \brr$.
\end{lemma}

\begin{proof}
    Denote $\bb \coloneqq P_\mathbb{B}$ and recall that $\brl p \brr = p$ on atoms. We prove the statement by induction on the rules in $\bb$.
    \begin{description}
        \item[Ax:] $L = p$, then $\llbracket p \rrbracket \subseteq \llbracket p \rrbracket$ is immediate.
        
        \item[Subs:] The rule is of shape 
        \begin{bprooftree}
        \def\ScoreOverhang{0.5pt}
            \AxiomC{$L_1 \vdash q$} 
            \AxiomC{$L_2;q \vdash p$} 
            \BinaryInfC{$L \vdash p$}
        \end{bprooftree}.
        By inductive hypothesis, $\brl L_1 \brr \subseteq q$ and $\brl L_2 \brr \otimes_M q \subseteq p$. Hence, $\brl L_2 \brr \otimes_M \brl L_1 \brr \subseteq p$,~i.e.~$\brl L \brr \subseteq p$.
        
        \item[Raa:] The rule is of shape 
        \begin{bprooftree}
        \def\ScoreOverhang{0.5pt}
            \AxiomC{$L; p^\bot \vdash \bot$}  
            \UnaryInfC{$L \vdash p$}
        \end{bprooftree}.
        By inductive hypothesis, $\brl L \brr \otimes_M \brl p^\bot \brr \subseteq \bot$. By Proposition~\ref{prop:psproperties}, this is equivalent to $\brl L \brr \subseteq \brl p^\bot \brr^\bot$, i.e.~$\brl L \brr \subseteq \brl p\brr^{\bot\bot}$. Hence, $\brl L \brr \subseteq p$ since $p$ is a fact.
        
        \item[$\otimes:$] The rule is of shape 
        \begin{bprooftree}
        \def\ScoreOverhang{0.5pt}
            \AxiomC{$L_1 \vdash {a_1}$} 
            \AxiomC{$\dots$} 
            \AxiomC{$L_{n-1} \vdash {a_{n-1}}$} 
            \AxiomC{$L_n;b_1,\dots,b_m \vdash \bot$} 
            \QuaternaryInfC{$L \vdash \bot$}
        \end{bprooftree}.
        By inductive hypothesis, $\brl L_i \brr \subseteq a_i$ for all $i$ and $\brl L_n \brr \otimes_M b_1 \otimes_M \dots \otimes_M b_m \subseteq \bot$. We also know that $\bigotimes_{i=1}^n a_i \subseteq \bigotimes_{i=j}^m b_j$. Hence $\bigotimes_{i=1}^{n-1} \brl L_i \brr \subseteq \bigotimes_{i=j}^m b_j$. Thus $\bigotimes_{i=1}^{n-1} \brl L_i \brr \otimes_M \brl L_n \brr \subseteq \bot$, i.e.~$\brl L \brr \subseteq \bot$.
        
        \item[$\cap:$] The rule is of shape 
        \begin{bprooftree}
        \def\ScoreOverhang{0.5pt}
            \AxiomC{$L_1 \vdash {a_1}$} 
            \AxiomC{$\dots$} 
            \AxiomC{$L_1 \vdash {a_n}$} 
            \AxiomC{$L_2;b \vdash \bot$} 
            \QuaternaryInfC{$L \vdash \bot$}
        \end{bprooftree}.
        By inductive hypothesis, $\brl L_1 \brr \subseteq a_i$ for all $i$ and $\brl L_2 \brr \otimes_M b \subseteq \bot$. We also know that $\bigcap_{i=1}^n a_i \subseteq b$. Then $\brl L_1 \brr \subseteq \bigcap_{i=1}^n a_i \subseteq b$. Hence, $\brl L_1 \brr \otimes_M \brl L_2 \brr \subseteq \bot$,~i.e.~$\brl L \brr \subseteq \bot$.

        \item[$\bang_1:$] The rule is of shape
        \begin{bprooftree}
        \def\ScoreOverhang{0.5pt}
            \AxiomC{$L \vdash b$}
            \AxiomC{$K \vdash \bot$}
            \BinaryInfC{$L, K \vdash \bot$}
        \end{bprooftree}.
        By inductive hypothesis, $\brl L \brr \subseteq b$ and $\brl K \brr \subseteq \bot$. We also know that $b = \bang a$ for some $a$. By soundness of phase semantics (specifically, weakening), we have that $\brl \bang a \brr \otimes_M \brl K \brr \subseteq \bot$. Hence, $b \otimes_M \brl K \brr \subseteq \bot$, and, finally, $\brl L \brr \otimes_M \brl K \brr \subseteq \bot$.

        \item[$\bang_2:$] The rule is of shape
        \begin{bprooftree}
        \def\ScoreOverhang{0.5pt}
            \AxiomC{$S \vdash a$}  
            \UnaryInfC{$S \vdash b$}
        \end{bprooftree}. 
        By inductive hypothesis, $\brl S \brr \subseteq a$. We also know that $b = \bang a$ and $S$ is structural, i.e.~each $s\in S$ is of shape $\bang m$ for some $m$ by Lemma~\ref{lmm:structatinb}. By soundness of phase semantics (specifically, promotion), if $\brl S \brr \subseteq a$ then $\brl S \brr \subseteq \brl \bang a \brr$, i.e.~$\brl S \brr \subseteq b$.
    \end{description}
\end{proof}

Finally, we show that any formula in $P_\mathbb{B}$ is intersupported with its ``atomic'' version, i.e.~its corresponding fact in $P$ (which is, in turn, an atom in $P_\mathbb{B}$).

\begin{lemma} \label{lmm:intersupport}
    Given $P \in \mathbb{P}$, $\phi \Vdash_{P_\mathbb{B}}^\varnothing \brl \phi \brr$ and $\brl \phi \brr \Vdash_{P_\mathbb{B}}^\varnothing \phi$. That is, $\phi \dashv \vdash_{P_\mathbb{B}} \brl \phi \brr$.
\end{lemma}

\begin{proof}
     Denote $\bb \coloneqq P_\mathbb{B}$. We prove the statement by induction on the structure of $\phi$. We will illustrate the proof for $(\otimes)$ and $(\quest)$. The remaining cases are similar and can be found in Appendix~\ref{app:proofs}.

     \begin{description}
         \item[$\phi = \alpha \otimes \beta:$] Suppose for an arbitrary $\bc \supseteq \bb$ and $L$ that $\Vdash_\bc^L \alpha \otimes \beta$. Since $\alpha \Vdash_\bb^\varnothing \brl \alpha \brr$ and $\beta \Vdash_\bb^\varnothing \brl \beta \brr$ by inductive hypothesis, obtain $\alpha,\beta \Vdash_\bb^\varnothing \brl \alpha \brr \otimes \brl \beta \brr$ by~$(\otimes I)'$~\cite[Theorem 4.2]{DBLP:journals/corr/abs-2504-08349}. Since also $\Vdash_\bc^L \alpha \otimes \beta$, obtain $\Vdash_\bc^L \brl \alpha \brr \otimes \brl \beta \brr$ by Lemma~\ref{lemma:generictensor}. Since $\brl \alpha \brr \otimes_M \brl \beta \brr \subseteq \brl \alpha \otimes \beta \brr$, we have that $\brl \alpha \brr, \brl \beta \brr \Vdash_\bb^\varnothing \brl \alpha \otimes \beta \brr$. By Lemma~\ref{lemma:generictensor}, obtain $\Vdash_\bc^L \brl \alpha \otimes \beta \brr$. Hence, by~\ref{eq:supp-inf}, $\alpha \otimes \beta \Vdash_\bb^\varnothing \brl \alpha \otimes \beta \brr$.

         Now, assume for an arbitrary $\bc \supseteq \bb$ and $L$ that $\Vdash_\bc^L \brl \alpha \otimes \beta \brr$. Further, for an arbitrary $\bd \supseteq \bc$ and $K$, assume that $\alpha,\beta \Vdash_\bd^K \bot$. By inductive hypothesis, $\brl \alpha \brr,\brl \beta \brr \Vdash_\bd^K \bot$. Since $\brl \alpha \otimes \beta \brr \subseteq \brl \alpha \brr \otimes_M \brl \beta \brr$, we have that $\brl \alpha \otimes \beta \brr \Vdash_\bb^\varnothing \brl \alpha \brr \otimes \brl \beta \brr$. By~\ref{eq:supp-inf}, obtain $\Vdash_\bc^L \brl \alpha \brr \otimes \brl \beta \brr$. Since also $\brl \alpha \brr,\brl \beta \brr \Vdash_\bd^K \bot$, obtain $\Vdash_\bd^{K,L} \bot$ by~\ref{eq:supp-tensor}. Since $\bd \supseteq \bc$ was such that $\alpha,\beta \Vdash_\bd^K \bot$, obtain $\Vdash_\bc^L \alpha \otimes \beta$ by~\ref{eq:supp-tensor}. By~\ref{eq:supp-inf} again, $\brl \alpha \otimes \beta \brr \Vdash_\bb^\varnothing \alpha \otimes \beta$.
    
         \item[$\phi = \quest \alpha:$] Suppose for an arbitrary $\bc \supseteq \bb$ and $L$ that $\Vdash_\bc^L \quest \alpha$. Notice that $\brl \bang(\alpha^\bot) \brr \subseteq \brl \alpha^\bot \brr$, hence $\brl \bang(\alpha^\bot) \brr \myv \brl \alpha^\bot \brr$, By inductive hypothesis, $\brl \bang(\alpha^\bot) \brr \myv \alpha ^\bot$,~i.e.~$\Vdash_\bb^{\brl \bang(\alpha^\bot) \brr} \alpha^\bot$. Since $\brl \bang (\alpha^\bot) \brr$ is structural in $\bb$ (see Lemma~\ref{lmm:structatinb}), obtain $\Vdash_\bc ^{L,\brl \bang(\alpha^\bot) \brr} \bot$,~i.e.~$\brl \bang(\alpha^\bot) \brr \Vdash_\bc ^{L} \bot$ by~\ref{eq:supp-quest}. By soundness of phase semantics, $\brl (\quest \alpha)^\bot \brr \subseteq \brl \bang(\alpha^\bot) \brr$, hence $\brl (\quest \alpha)^\bot \brr \myv \brl \bang(\alpha^\bot) \brr$. By~\ref{eq:supp-inf}, $\brl (\quest \alpha)^\bot \brr \Vdash_\bc ^{L} \bot$, and by \textit{reductio ad absurdum}, $\Vdash_\bc ^{L} \brl \quest \alpha \brr$. By~\ref{eq:supp-inf} then, $\quest \alpha \Vdash_\bb^\varnothing \brl \quest \alpha \brr$.
        
         Now, assume for an arbitrary $\bc \supseteq \bb$ and $L$ that $\Vdash_\bc^L \brl \quest \alpha \brr$. Further assume, for an arbitrary $\bd \supseteq \bc$ and structural $S$, that $\Vdash_\bd^S \alpha^\bot$. By inductive hypothesis, $\Vdash_\bd^S \brl \alpha^\bot \brr$. Since there is a rule in $\bb$ concluding $S \vdash \brl \bang (\alpha^\bot) \brr$  from $S \vdash \brl \alpha^\bot \brr$, we obtain $\Vdash_\bd^S \brl \bang (\alpha^\bot) \brr$. By soundness of phase semantics, $\brl \bang(\alpha^\bot) \brr \subseteq \brl (\quest \alpha)^\bot \brr$, hence $\brl \bang(\alpha^\bot) \brr \myv \brl (\quest \alpha)^\bot \brr$. By~\ref{eq:supp-inf}, $\Vdash_\bd^S \brl (\quest \alpha)^\bot \brr$, and since $\brl \quest \alpha \brr, \brl (\quest \alpha)^\bot \brr \myv \bot$, by \ref{eq:supp-inf} again, $\Vdash_\bd^{S,L} \bot$. Since $\bd \supseteq \bc$ was such that $\Vdash_\bd^S \alpha^\bot$, conclude $\Vdash_\bc^L \quest \alpha$ by~\ref{eq:supp-quest}. By~\ref{eq:supp-inf} then, $\brl \quest \alpha \brr \Vdash_\bb^\varnothing \quest \alpha$.
    \end{description}
\end{proof}

\subsection{From Base-extension semantics to Phase semantics}
Now that we know how to construct a base from a phase model, we wish to do the opposite. We hence define a map from the set of bases to the set of phase models.

\begin{definition}[Base to Phase] \label{def:basetophase}
    We define the map $(-)_\mathbb{P}: \bs \rightarrow \ps$ as $({\At}_\bb, {\Rule}_\bb, \flat) = \bb \mapsto \bb_\mathbb{P} =(D_{M}, I, \bot, v)$, where:
    \begin{itemize}
        \item $(M,\uplus,\varnothing)$ is the monoid of multisets of formulas on $\At_\bb$ (i.e.~$\uplus$ is multiset union),
        \item $J \coloneqq \{\bang \Gamma \; | \; \Gamma \text{ is a multiset of formulas} \}$,
        \item $\bot \coloneqq \{\Gamma \; | \; \Gamma \Vdash_\bb^\varnothing \bot\}$,
        \item $v$ is given by $a \mapsto \{\Gamma \; | \; \Gamma \Vdash_\bb^\varnothing \flat(a) \}$.
    \end{itemize}
\end{definition}

Note that $\bang \phi \myv 1$ for any $\phi$ (see Theorem~\ref{thm:soundness}), hence $J \subseteq 1$ by Lemma~\ref{lmm:validphasespace}; hence $I = J \cap 1 = J$.

Every $X \subseteq M$ of the shape $X = \{\Gamma \; | \; \Gamma \Vdash_\bb^\varnothing \phi\}$ for some formula $\phi \in {\forms}_\bb$ is a fact in $\bb_\mathbb{P}$: observe that $X^{\bot\bot} = \{\Theta \; | \; \Theta,\Delta \myv \bot \text{ for all } \Delta \text{ satisfying } \Gamma \myv \phi \Rightarrow \Gamma,\Delta \myv \bot \text{ for } \Gamma \in X\}$. Thus, if $\Sigma \in X$, i.e.~$\Sigma \myv \phi$, then $\Sigma,\Delta \myv \bot$, hence $\Sigma \in X^{\bot\bot}$. Conversely, if $\Sigma \in X^{\bot\bot}$, choose $\Delta =  \{\phi^\bot\}$. Since $\phi, \phi^\bot \myv \bot$, we have that $\Gamma \myv \phi \Rightarrow \Gamma,\phi^\bot \myv \bot$, hence we can conclude $\Sigma, \phi^\bot \myv \bot$. By Lemma~\ref{lmm:raa}, $\Sigma \myv \phi$, i.e.~$\Sigma \in X$. Thus, $X = X^{\bot\bot} \in \bb_\mathbb{P}$.

\begin{example}
    Let $\bb$ be a base over atoms $\{p,q,\bot\}$, whose valuation $\flat$ is a surjection, with the following set of rules: $$\left\{\begin{bprooftree}
            \AxiomC{}
            \UnaryInfC{$\vdash p$}
        \end{bprooftree}\right\} \cup \left\{\begin{bprooftree}
            \AxiomC{$L \vdash q$}
            \UnaryInfC{$L \vdash p$}
        \end{bprooftree}\right\}_L$$
Then $\bb_\mathbb{P}$ is the monoid of multisets of linear logic formulas inductively defined on $\{p,q,\bot\}$ with a fixed subset given by $\bot_\mathbb{P} = \{\Gamma \; | \; \Gamma \Vdash_\bb \bot\}$ and a valuation $v$ given by $a \mapsto \{\Gamma \; | \; \Gamma \Vdash_\bb^\varnothing \flat(a) \}$, where $\flat(a) \in \{p,q,\bot\}$. By Lemma~\ref{lmm:validphasespace} it will follow that every definable fact of $\bb_\mathbb{P}$ is of shape $\{\Gamma \; | \; \Gamma \Vdash_\bb \phi\}$ for some formula $\phi$. Below we list some examples of definable facts:
    \begin{itemize}
        \item $\{\Gamma \; | \; \Gamma \Vdash_\bb p\} = \{\varnothing,\{p\},\{q\},\{p\with\bot\},\dots\}$
        \item $\{\Gamma \; | \; \Gamma \Vdash_\bb q\} = \{\{q\},\{p, p^\bot \parr q\},\dots\}$
        \item $\{\Gamma \; | \; \Gamma \Vdash_\bb p \otimes q\} = \{\{q\},\{p,q\},\{p \otimes q\},\dots\}$
    \end{itemize}
\end{example}

Formulas over a base can be interpreted in the induced phase space: the function $(-)^*: \forms_\bb \rightarrow {\Fact}_{\bb_\mathbb{P}}$ is defined on atoms by $p \mapsto \{\Gamma \; | \; \Gamma \Vdash_\bb^\varnothing p \}$ and inductively extended to all formulas by interpreting the connectives in the phase space in the usual way (cf.~Definition~\ref{def:interpretation}). This gives us an alternative definition of a valuation function: $v(a) = (\flat(a))^*$.

From the following lemma, we conclude that the equality persists to the interpretations, i.e.~extensions of $\flat$ and $v$. 

\begin{lemma}
    Let $A \in {\forms}_{\At}$. For every base valuation $\flat: {\At} \rightarrow \At_\bb$ and the induced phase model valuation $v: {\At} \rightarrow {\Fact}_{\bb_\mathbb{P}}$ we have $A^v = (A^\flat)^*$.
\end{lemma}

\begin{proof}
    We prove the statement by induction on the structure of $A$, with $\odot \in \{\otimes,\parr,\with,\oplus\}$.
    \begin{description}
        \item[$A = p:$] $v(a) = (\flat(a))^*$ by definition.

        \item[$A \in \{\bot, 1,\top,0\}:$] $A^v = (A^\flat)^*$ since $A^\flat = A$ and $(-)^*$ and $(-)^v$ have the same definition on units (Definition~\ref{def:interpretation}).

        \item[$A = \alpha^\bot:$] $(\alpha^\bot)^v = (\alpha^v)^\bot$. Now, by inductive hypothesis $(\alpha^v)^\bot = ((\alpha^\flat)^*)^\bot = ((\alpha^\flat)^\bot)^* = ((\alpha^\bot)^\flat)^*$.

        \item[$A = \alpha \odot \beta:$] $(\alpha \odot \beta)^v = (\alpha)^v \odot_{\bb_\mathbb{P}} (\beta)^v$. Now, by inductive hypothesis $(\alpha)^v \odot_{\bb_\mathbb{P}} (\beta)^v = ((\alpha)^\flat)^* \odot_{\bb_\mathbb{P}} ((\beta)^\flat)^* = ((\alpha)^\flat \odot_{\bb} (\beta)^\flat)^* = ((\alpha \odot \beta)^\flat)^*$.

        \item[$A = \bang \alpha:$] $(\bang \alpha)^v = \bang_{\bb_\mathbb{P}} (\alpha)^v$. Now, by inductive hypothesis $\bang_{\bb_\mathbb{P}} (\alpha)^v = \bang_{\bb_\mathbb{P}} ((\alpha)^\flat)^* = (\bang_{\bb} (\alpha)^\flat)^* = ((\bang \alpha)^\flat)^*$.

        \item[$A = \quest \alpha:$] $(\quest \alpha)^v = \quest_{\bb_\mathbb{P}} (\alpha)^v$. Now, by inductive hypothesis $\quest_{\bb_\mathbb{P}} (\alpha)^v = \quest_{\bb_\mathbb{P}} ((\alpha)^\flat)^* = (\quest_{\bb} (\alpha)^\flat)^* = ((\quest \alpha)^\flat)^*$.
    \end{description}
\end{proof}

Next, we show that facts obtained via operations in the phase space correspond to facts induced by formulas, e.g.~$\{\Gamma \; | \; \Gamma \Vdash_\bb^\varnothing \alpha \} \parr_M  \{\Gamma \; | \; \Gamma \Vdash_\bb^\varnothing \beta \}  = \{\Gamma \; | \; \Gamma \Vdash_\bb^\varnothing \alpha \parr \beta\}$.

\begin{lemma} \label{lmm:validphasespace}
    For every base $\bb$ and formula $\phi \in \forms_\bb$ we have $\phi^* = \{\Gamma \; | \; \Gamma \Vdash_\bb^\varnothing \phi \}$ in $\bb_\mathbb{P}$.
\end{lemma}

\begin{proof}
     We prove the statement by induction on the structure of $\phi$. We will illustrate the proof for $(\parr)$ and $(\bang)$. The remaining cases are similar and can be found in Appendix~\ref{app:proofs}.

     \begin{description}
         \item[$\phi = \alpha \parr \beta:$] Suppose $\Gamma \in (\alpha \parr \beta)^* = (\alpha)^* \parr_M (\beta)^* = \{\Gamma \; | \; \Gamma \Vdash_\bb^\varnothing \alpha \} \parr_M  \{\Gamma \; | \; \Gamma \Vdash_\bb^\varnothing \beta \} = (\{\Gamma \; | \; \Gamma \Vdash_\bb^\varnothing \alpha \}^\bot \uplus \{\Gamma \; | \; \Gamma \Vdash_\bb^\varnothing \beta \}^\bot)^{\bot} = (\{\Gamma \; | \; \Gamma \Vdash_\bb^\varnothing \alpha^\bot\} \uplus \{\Gamma \; | \; \Gamma \Vdash_\bb^\varnothing \beta^\bot \} \brr)^{\bot}$ by inductive hypothesis. Hence, $\Gamma,\Delta,\Sigma \myv \bot$ whenever $\Delta \myv \alpha^\bot$ and $\Sigma \myv \beta^\bot$. Choose $\Delta = \{\alpha^\bot\}$ and $\Sigma = \{\beta^\bot\}$. Hence, $\Gamma, \alpha^\bot, \beta^\bot \myv \bot$. By~$(\parr I)'$~\cite[Theorem 4.2]{DBLP:journals/corr/abs-2504-08349}, $\Gamma \myv \alpha \parr \beta$, hence $\Gamma \in \{\Gamma \; | \; \Gamma \Vdash_\bb^\varnothing \alpha \parr \beta\}$.

         Conversely, suppose that $\Gamma \in \{\Gamma \; | \; \Gamma \Vdash_\bb^\varnothing \alpha \parr \beta \}$ and assume that $\Delta \in \{\Gamma \; | \; \Gamma \Vdash_\bb^\varnothing \alpha^\bot \}$ and $\Sigma \in \{\Gamma \; | \; \Gamma \Vdash_\bb^\varnothing \beta^\bot \}$,~i.e.~that $\Delta \myv \alpha^\bot$ and $\Sigma \myv \beta^\bot$. Now, assume, for an arbitrary $\bc \supseteq \bb$ and atomic multisets $L,K,M$, that $\Vdash_\bc^L \Gamma$, $\Vdash_\bc^K \Delta$ and $\Vdash_\bc^M \Sigma$. Hence, $\Vdash_\bc^L \alpha \parr \beta$ and $\Vdash_\bc^K \alpha^\bot$ and $\Vdash_\bc^M \beta^\bot$. By~\ref{eq:supp-parr}, conclude $\Vdash_\bc^{L,K,M} \bot$. Hence, by~\ref{eq:supp-inf}, $\Gamma,\Delta,\Sigma \myv \bot$, thus $\Gamma \in (\{\Gamma \; | \; \Gamma \Vdash_\bb^\varnothing \alpha^\bot\} \uplus \{\Gamma \; | \; \Gamma \Vdash_\bb^\varnothing \beta^\bot \}r)^{\bot}$. Hence, $\Gamma \in (\{\Gamma \; | \; \Gamma \Vdash_\bb^\varnothing \alpha \}^\bot \uplus \{\Gamma \; | \; \Gamma \Vdash_\bb^\varnothing \beta \}^\bot)^{\bot} = \{\Gamma \; | \; \Gamma \Vdash_\bb^\varnothing \alpha \} \parr_M  \{\Gamma \; | \; \Gamma \Vdash_\bb^\varnothing \beta \} = (\alpha)^* \parr_M (\beta)^*$ by inductive hypothesis,~i.e.~$\Gamma \in (\alpha \parr \beta)^*$.

        \item[$\phi = \bang \alpha:$] Suppose $\Gamma \in (\bang \alpha)^* = \bang_M (\alpha)^* = \bang_M \{\Gamma \; | \; \Gamma \Vdash_\bb^\varnothing \alpha \} = (\{\Gamma \; | \; \Gamma \Vdash_\bb^\varnothing \alpha \} \cap I)^{\bot\bot}$ by inductive hypothesis. Hence, $\Gamma, \Delta \myv \bot$ for all $\Delta$ s.t. $\Sigma,\Delta \myv \bot$ whenever $\Sigma \myv \alpha$ and $\Sigma \in I$. Notice that, if $\Sigma \in I$, then every $\phi \in \Sigma$ is of the shape $\bang \psi$ for some $\psi$. Hence, if $\Sigma \myv \alpha$, then $\Sigma \myv \bang \alpha$ by~$(\bang I)'$ (Theorem~\ref{thm:soundness}). Further notice that $\bang \alpha, (\bang \alpha)^\bot \myv \bot$, thus $\Sigma, (\bang \alpha)^\bot \myv \bot$. Now, choose $\Delta = \{(\bang \alpha)^\bot\}$, so that we can conclude $\Gamma,(\bang \alpha)^\bot \myv \bot$. By Lemma~\ref{lmm:raa}, $\Gamma \myv \bang \alpha$, hence $\Gamma \in \{\Gamma \; | \; \Gamma \Vdash_\bb^\varnothing \bang \alpha \}$.

        Conversely, suppose that $\Gamma \in \{\Gamma \; | \; \Gamma \Vdash_\bb^\varnothing \bang \alpha \}$,~i.e.~$\Gamma \myv \bang \alpha$, and assume that $\Delta \in (\{\Gamma \; | \; \Gamma \Vdash_\bb^\varnothing \alpha \} \cap I)^\bot$,~i.e.~that $\Delta, \Sigma \myv \bot$ whenever $\Sigma \myv \alpha$ and $\Sigma \in I$. Now, choose $\Sigma = \{\bang \alpha\}$. Then $\Delta, \bang \alpha \myv \bot$, and since $\Gamma \myv \bang \alpha$, obtain $\Delta,\Gamma \myv \bot$. Hence, $\Gamma \in (\{\Gamma \; | \; \Gamma \Vdash_\bb^\varnothing \alpha \} \cap I)^{\bot\bot} = \bang_M \{\Gamma \; | \; \Gamma \Vdash_\bb^\varnothing \alpha \}r = \bang_M (\alpha)^*$ by inductive hypothesis. Thus, $\Gamma \in (\bang \alpha)^*$.
     \end{description}
\end{proof}

Recall that every definable fact $X \in D_M$ of $\bb_\mathbb{P}$ is of shape $A^v$ for some $A \in {\forms_{\At}}$ (cf.~Definition~\ref{def:definablefact}). Finally, we conclude that every such fact $A^v \in \bb_\mathbb{P}$ is exactly $\{\Gamma \; | \; \Gamma \Vdash_\bb^\varnothing A^\flat \}$. 

\begin{corollary}
    For every base $\bb$ and $A^v \in \bb_\mathbb{P}$ we have that $A^v = (A^\flat)^* = \{\Gamma \; | \; \Gamma \Vdash_\bb^\varnothing A^\flat \}$.
\end{corollary}

\begin{remark}
    Note that the choice of $J$ affects the definable facts in a given phase model: have we chosen $J \coloneqq \{\varnothing\}$ in Definition~\ref{def:basetophase}, a fact $\{\Gamma \; | \; \Gamma \Vdash_\bb^\varnothing \bang \phi \}$ would no longer be definable for those $\phi$ which are not valid in $\bb$, i.e.~if $\varnothing \notin \{\Gamma \; | \; \Gamma \Vdash_\bb^\varnothing \phi \}$. In other words, Theorem~\ref{lmm:validphasespace} would fail: $\bang_M \{\Gamma \; | \; \Gamma \Vdash_\bb^\varnothing \phi \} = (\{\Gamma \; | \; \Gamma \Vdash_\bb^\varnothing \phi \} \cap I)^{\bot\bot} = (\{\Gamma \; | \; \Gamma \Vdash_\bb^\varnothing \phi \} \cap \{\varnothing\})^{\bot\bot} = \varnothing^{\bot\bot} = \{\Gamma \; | \; \Gamma \in M \}^{\bot} = \{\Gamma \; | \; \Gamma \Vdash_\bb^\varnothing \top \}^{\bot} = \{\Gamma \; | \; \Gamma \Vdash_\bb^\varnothing 0 \} \neq \{\Gamma \; | \; \Gamma \Vdash_\bb^\varnothing \bang \phi \}$.
\end{remark}

\subsection{...and Back}
We now proceed to establish the equivalence between the two semantics. To do so, we first have to define a notion of a structure-preserving map between phase models. We call such a map a \emph{phase model morphism}:

\begin{definition}[Phase model morphism]
    We define a \emph{phase model morphism} $p:(D_M,I_M,\bot_M, v_M) \rightarrow (D_N,I_N,\bot_N, v_N)$ as a function $p: D_M \rightarrow {\Fact}_N$ that preserves the valuation: $p \circ v_M = v_N$; the inclusion order: $\alpha \subseteq \beta$ implies $p(\alpha) \subseteq p(\beta)$; and tensor unit: $p(1_M) = 1_N$.
\end{definition}
In particular, phase model morphisms preserve validity: $1_M \subseteq \beta$ implies $p(1_M) = 1_N \subseteq p(\beta)$.

\begin{remark}
    Note that the set of all facts of a phase space can be viewed as a residuated lattice. The phase model morphism we define does not claim to preserve the full algebraic structure of the lattice -- for the purposes of this work, the preservation of valuation, inclusion order and tensor unit suffices. 
\end{remark}

\begin{definition} \label{def:monoidmap}
    Let $P = (D_M, I, \bot,v)$ be a phase model. We define the map  $f_P: P \rightarrow (P_{\mathbb{B}})_{\mathbb{P}}$ as $m \mapsto \{\Gamma \; | \; \Gamma \Vdash_{P_{\mathbb{B}}}^\varnothing m\}$, and the map $g_P : (P_{\mathbb{B}})_{\mathbb{P}} \rightarrow P$ as $\alpha \mapsto \bigcup_{L\in\alpha} \brl L \brr$.
\end{definition}

Notice that $f_P(m) = m^*$.

\begin{remark}
    Note that, perhaps counterintuitively, the definition of $g_P$ could, but does not need to be $\alpha \mapsto (\bigcup_{L\in\alpha} \brl L \brr)^{\bot\bot}$. That is due to the nature of definable facts of $(P_{\mathbb{B}})_{\mathbb{P}}$: recall that all of them are of shape $\alpha = \{\Gamma \; | \; \Gamma \Vdash_{P_{\mathbb{B}}}^\varnothing \phi\}$ for some $\phi$. The purpose of $g_P$ is to map such $\alpha$ to $\brl \phi \brr \in D_P$. Since $\brl L \brr \subseteq \brl \phi \brr$ for all $L \in \alpha$ by Lemma~\ref{lmm:inclusion}, we are guaranteed that $\bigcup_{L\in\alpha} \brl L \brr \subseteq \brl \phi \brr$. At last, since $\brl \phi \brr \dashv \vdash_\bb \phi$ by Lemma~\ref{lmm:intersupport}, we know that there exists $L = \{\brl \phi \brr \}$, which turns an inclusion into an equality: $\bigcup_{L\in\alpha} \brl L \brr = \brl \phi \brr$.
\end{remark}

Note that $f_P$ and $g_P$ are phase model morphisms. The inclusion order is preserved: let $m,n \in P$ s.t.~$m \subseteq n$. Then $m \Vdash^\varnothing_{P_\mathbb{B}} n$ since $P_\mathbb{B}$ preserves inclusion by construction (Definition~\ref{def:phasetobase}). Now, suppose $\Delta \in  f_P(m) = \{\Gamma \; | \; \Gamma \Vdash_{P_{\mathbb{B}}}^\varnothing m\}$. Then $\Delta \Vdash_{P_{\mathbb{B}}}^\varnothing m$ and, by~\ref{eq:supp-inf}, $\Delta \Vdash_{P_{\mathbb{B}}}^\varnothing n$, hence $\Delta \in \{\Gamma \; | \; \Gamma \Vdash_{P_{\mathbb{B}}}^\varnothing n\} = f_P(n)$, i.e.~$f_P(m) \subseteq f_P(n)$. Similarly, let $\alpha, \beta \in (P_{\mathbb{B}})_{\mathbb{P}}$ s.t.~$\alpha \subseteq \beta$. That is, $\{\Gamma \; | \; \Gamma \Vdash_{P_{\mathbb{B}}}^\varnothing A\} \subseteq \{\Gamma \; | \; \Gamma \Vdash_{P_{\mathbb{B}}}^\varnothing B\}$. Thus, $L \in \alpha$ implies $L \in \beta$ for any atomic $L$. Hence, $\bigcup_{L\in g_P(\alpha)} \brl L \brr \subseteq \bigcup_{L\in g_P(\beta)} \brl L \brr$. That is, $g_P(\alpha) \subseteq g_P(\beta)$.

The tensor unit is preserved: $f_P(1_P) = \{\Gamma \; | \; \Gamma \Vdash_{P_{\mathbb{B}}}^\varnothing 1_P\} = \{\Gamma \; | \; \Gamma \Vdash_{P_{\mathbb{B}}}^\varnothing \bot^\bot\} = \{\Gamma \; | \; \Gamma \Vdash_{P_{\mathbb{B}}}^\varnothing \bot\}^\bot = 1_{(P_{\mathbb{B}})_{\mathbb{P}}}$ by Lemmas~\ref{lmm:intersupport} and~\ref{lmm:validphasespace}. Similarly, $g_P(1_{(P_{\mathbb{B}})_{\mathbb{P}}}) = g_P(\{\Gamma \; | \; \Gamma \Vdash_{P_{\mathbb{B}}}^\varnothing \bot\}^\bot) = g_P(\{\Gamma \; | \; \Gamma \Vdash_{P_{\mathbb{B}}}^\varnothing 1_P\}) = \bigcup_{L\in \{L \; | \; L \vdash_{P_{\mathbb{B}}} 1_P\}} \brl L \brr = 1_P$ since $1_P \vdash_{P_{\mathbb{B}}} 1_P$ and by Lemma~\ref{lmm:deriveqsupport} and~\ref{lmm:inclusion} .

The valuation is preserved: notice that $f_P \circ v$ is given by $a \mapsto (v(a))^* = \{\Gamma \; | \; \Gamma \Vdash_\bb^\varnothing v(a) \}$. Further, by Definitions~\ref{def:phasetobase} and~\ref{def:basetophase}, the valuation function $w: {\At} \rightarrow {\Fact}_{(P_{\mathbb{B}})_{\mathbb{P}}}$ is given by $a \mapsto \{\Gamma \; | \; \Gamma \Vdash_\bb^\varnothing \flat(a) \} = \{\Gamma \; | \; \Gamma \Vdash_\bb^\varnothing v(a) \}$. Hence, $f_P \circ v = w$. On the other hand, $g_P \circ w$ is given by $a \mapsto w(a) = \{\Gamma \; | \; \Gamma \Vdash_\bb^\varnothing v(a) \} \mapsto \bigcup_{L\in w(a)} \brl L \brr$. This is equal to $\bigcup_{L\in \{L \; | \; L \vdash_{P_{\mathbb{B}}} v(a)\}} \brl L \brr$ by Lemma~\ref{lmm:deriveqsupport}. Therefore $\bigcup_{L\in w(a)} \brl L \brr \subseteq v(a)$ by Lemma~\ref{lmm:inclusion}, and since $\{v(a)\} \in w(a)$, conclude that $\bigcup_{L\in w(a)} \brl L \brr =v(a)$. Hence, $g_P \circ w = v$.

Now, we show that the composition of these phase model morphisms yields an identity map, hence proving the desired equivalence.

\begin{theorem}
    For every $P \in \mathbb{P}$ we have $g_P \circ f_P = Id_P$ and $f_P \circ g_P = Id_{(P_{\mathbb{B}})_{\mathbb{P}}}$.
\end{theorem}

\begin{proof}
    Let $m \in P$. Then $f_P(m) = \{\Gamma \; | \; \Gamma \Vdash_{P_{\mathbb{B}}}^\varnothing m\}$ and $g_P \circ f_P(m) = \bigcup_{L\in f_P(m)} \brl L \brr$. This is equal to $\bigcup_{L\in \{L \; | \; L \vdash_{P_{\mathbb{B}}} m\}} \brl L \brr$ by Lemma~\ref{lmm:deriveqsupport}. Therefore $\bigcup_{L\in f_P(m)} \brl L \brr \subseteq m$ by Lemma~\ref{lmm:inclusion}, and since $\{m\} \in f_P(m)$, conclude that $\bigcup_{L\in f_P(m)} \brl L \brr = g_P \circ f_P(m) = m$.

    Now, let $\alpha = \{\Gamma \; | \; \Gamma \Vdash_{P_{\mathbb{B}}}^\varnothing A\} \in (P_{\mathbb{B}})_{\mathbb{P}}$. Then $g_P(\alpha) = \bigcup_{L\in\alpha} \brl L \brr$ and $f_P \circ g_P(\alpha) = \{\Gamma \; | \; \Gamma \Vdash_{P_{\mathbb{B}}}^\varnothing g_P(\alpha)\}$. Notice that by Lemma~\ref{lmm:inclusion}, $g_P(\alpha) \subseteq \brl A \brr$, and since $\brl A \brr \in \alpha$ yields $\brl A \brr \subseteq g_P(\alpha)$, we conclude  $g_P(\alpha) = \brl A \brr$. Now, $\Gamma \in \alpha \Leftrightarrow \Gamma \Vdash_{P_{\mathbb{B}}}^\varnothing A \Leftrightarrow \Gamma \Vdash_{P_{\mathbb{B}}}^\varnothing \brl A \brr$ by Lemma~\ref{lmm:intersupport}, and $\Gamma \Vdash_{P_{\mathbb{B}}}^\varnothing \brl A \brr \Leftrightarrow \Gamma \Vdash_{P_{\mathbb{B}}}^\varnothing g_P(\alpha) \Leftrightarrow \Gamma \in f_P \circ g_P(\alpha)$. Hence $f_P \circ g_P(\alpha) = \alpha$.
\end{proof}

Next, we define a \emph{base morphism} -- a map between bases that preserves the support relation.  

\begin{definition}[Base morphism]
    We define a \emph{base morphism} $b:({\At}_\bb, {\Rule}_\bb, \sigma) \rightarrow ({\At}_\bc, {\Rule}_\bc, \tau)$ as a function $b: \sigma({\At}) \rightarrow {\At}_\bc$ such that the inductively induced map on formulas $b: ({\forms}_{\At})^\sigma \rightarrow \forms_\bc$ preserves the valuation: $b \circ \sigma(a) \dashv \vdash_{\bc} \tau(a)$; and the support relation: $\myv \phi$ implies $\Vdash_\bc^\varnothing b(\phi)$.
\end{definition}

\begin{definition} \label{def:basemap}
    Let $\bb = ({\At}_{\bb}, {\Rule}_{\bb}, \sigma)$ be a base. We define the map $f_B : \bb \rightarrow (\bb_\mathbb{P})_\mathbb{B}$ as $p \mapsto p^*$, and a map $g_B: (\bb_\mathbb{P})_\mathbb{B} \rightarrow \bb$ by mapping each atom $q \in {(\sigma_\mathbb{P})_\mathbb{B}}({\At})$ to some atom $p \in \sigma({\At})$ such that $p^* = \brl q \brr$.
\end{definition}

We next demonstrate that $g_B$ acts on formulas in line with its definition on atoms.

\begin{lemma} \label{lmm:starissquarebracket}
    Let $\bb = ({\At}_{\bb}, {\Rule}_{\bb}, \sigma)$ be a base and $\phi \in (\forms_{\At})^{(\sigma_\mathbb{P})_\mathbb{B}}$. Then $ (g_B(\phi))^* = \brl \phi \brr$.
\end{lemma}

\begin{proof}
    We prove the statement by induction on the structure of $\phi$, with $\odot \in \{\otimes,\parr,\with,\oplus\}$.
    \begin{description}
        \item[$\phi = p:$] $g_B(p) = q$ s.t. $q^* = p$, thus $(g_B(p))^* = q^* = \brl p \brr$ by definition.

        \item[$\phi \in \{\bot, 1,\top,0\}:$] $g_B(\phi) = \psi$ s.t.~$\psi^* = \brl \phi \brr$ since $(-)^{{(\sigma_\mathbb{P})_\mathbb{B}}}$ is identity on units. Hence, $(g_B(\phi))^* = \psi^* = \brl \phi \brr$.

        \item[$\phi = \alpha^\bot:$] $(g_B(\alpha^\bot))^* = (g_B(\alpha)^\bot)^* = (g_B(\alpha)^*)^\bot$. Now, by inductive hypothesis $(g_B(\alpha)^*)^\bot = \brl \alpha \brr^\bot = \brl \alpha^\bot \brr$.

        \item[$\phi = \alpha \odot \beta:$] $(g_B(\alpha \odot \beta))^* = (g_B(\alpha) \odot g_B(\beta))^* = (g_B(\alpha))^* \odot_{\bb_\mathbb{P}} (g_B(\beta))^*$. Now, by inductive hypothesis $(g_B(\alpha))^* \odot_{\bb_\mathbb{P}} (g_B(\beta))^* = \brl \alpha \brr \odot_{\bb_\mathbb{P}} \brl \beta \brr = \brl \alpha \odot \beta \brr$.

        \item[$\phi = \bang \alpha:$] $(g_B(\bang \alpha))^* = (\bang g_B(\alpha))^* = \bang_{\bb_\mathbb{P}} (g_B(\alpha))^*$. By inductive hypothesis $\bang_{\bb_\mathbb{P}} (g_B(\alpha))^* = \bang_{\bb_\mathbb{P}} \brl \alpha \brr = \brl \bang \alpha \brr$.

        \item[$\phi = \quest \alpha:$] $(g_B(\quest \alpha))^* = (\quest g_B(\alpha))^* = \quest_{\bb_\mathbb{P}} (g_B(\alpha))^*$. By inductive hypothesis $\quest_{\bb_\mathbb{P}} (g_B(\alpha))^* = \quest_{\bb_\mathbb{P}} \brl \alpha \brr = \brl \quest \alpha \brr$.
    \end{description}
\end{proof}

Note that the inductive extension of $f_B$ does not simplify to $(-)^*$, unlike in phase models (Lemma~\ref{lmm:validphasespace}); i.e.,~it is not the case that $f_B(\phi) = \phi^*$ for  $\phi \in (\forms_{\At})^\sigma$. However, it is the case that $\brl f_B(\phi) \brr = \phi^*$. The proof is identical to that of Lemma~\ref{lmm:starissquarebracket}, with the base case $\brl f_B(p) \brr = \brl p^* \brr = \brl \{\Gamma \; | \; \Gamma \Vdash_\bb^\varnothing p \} \brr = \{\Gamma \; | \; \Gamma \Vdash_\bb^\varnothing p \} = p^*$.

\begin{lemma} \label{lmm:starissquarebracket2}
    Let $\bb = ({\At}_{\bb}, {\Rule}_{\bb}, \sigma)$ be a base and $\phi \in (\forms_{\At})^\sigma$. Then $\brl f_B(\phi) \brr = \phi^*$.
\end{lemma}

\begin{proof}
    We prove the statement by induction on the structure of $\phi$, with $\odot \in \{\otimes,\parr,\with,\oplus\}$.
    \begin{description}
        \item[$\phi = p:$] $\brl f_B(p) \brr = \brl p^* \brr = \brl \{\Gamma \; | \; \Gamma \Vdash_\bb^\varnothing p \} \brr = \{\Gamma \; | \; \Gamma \Vdash_\bb^\varnothing p \} = p^*$.

        \item[$\phi \in \{\bot, 1,\top,0\}:$] $f_B(\phi) = \phi$ since $(-)^\sigma$ is identity on units. Hence, $\brl f_B(\phi) \brr = \brl \phi \brr =  \{\Gamma \; | \; \Gamma \Vdash_\bb^\varnothing \phi \}  = \phi^*$.

        \item[$\phi = \alpha^\bot:$] $\brl f_B(\alpha^\bot) \brr = \brl f_B(\alpha)^\bot \brr = \brl f_B(\alpha)\brr^\bot$. Now, by inductive hypothesis $\brl f_B(\alpha)\brr^\bot = (\alpha^*)^\bot = (\alpha^\bot)^*$.

        \item[$\phi = \alpha \odot \beta:$] $\brl f_B(\alpha \odot \beta) \brr = \brl f_B(\alpha) \odot f_B(\beta) \brr = \brl f_B(\alpha) \brr \odot_{\bb_\mathbb{P}} \brl f_B(\beta) \brr$. Now, by inductive hypothesis $\brl f_B(\alpha) \brr \odot_{\bb_\mathbb{P}} \brl f_B(\beta) \brr = \alpha^* \odot_{\bb_\mathbb{P}} \beta^* = (\alpha \odot \beta)^*$.

        \item[$\phi = \bang \alpha:$] $\brl f_B(\bang \alpha) \brr = \brl \bang f_B(\alpha) \brr = \bang_{\bb_\mathbb{P}} \brl f_B(\alpha) \brr$. By inductive hypothesis $\bang_{\bb_\mathbb{P}} \brl f_B(\alpha) \brr = \bang_{\bb_\mathbb{P}} (\alpha^*) = (\bang \alpha)^*$.

        \item[$\phi = \quest \alpha:$] $\brl f_B(\quest \alpha) \brr = \brl \quest f_B(\alpha) \brr = \quest_{\bb_\mathbb{P}} \brl f_B(\alpha) \brr$. By inductive hypothesis $\quest_{\bb_\mathbb{P}} \brl f_B(\alpha) \brr = \quest_{\bb_\mathbb{P}} (\alpha^*) = (\quest \alpha)^*$.
    \end{description}
\end{proof}

Note that $f_B$ and $g_B$ are base morphisms. The support relation is preserved: let $\phi \in (\forms_{\At})^\sigma$ s.t. $\myv \phi$. Then $\varnothing \in \phi^* = f_B(\phi)$, hence $\Vdash^\varnothing_{(\bb_\mathbb{P})_\mathbb{B}} f_B(\phi)$. Similarly, let $\psi \in (\forms_{\At})^{(\sigma_\mathbb{P})_\mathbb{B}}$ s.t. $\Vdash_{(\bb_\mathbb{P})_\mathbb{B}}^\varnothing \psi$. By Lemma~\ref{lmm:intersupport}, $\psi \Vdash_{(\bb_\mathbb{P})_\mathbb{B}}^\varnothing \brl \psi \brr$, hence  $\Vdash_{(\bb_\mathbb{P})_\mathbb{B}}^\varnothing \brl \psi \brr$. By Lemmas~\ref{lmm:deriveqsupport},~\ref{lmm:inclusion} and~\ref{lmm:starissquarebracket}, $\varnothing \in \brl \psi \brr = (g_B(\psi))^* = \{\Gamma \; | \; \Gamma \Vdash_\bb^\varnothing g_B(\psi)\}$ by Lemma~\ref{lmm:validphasespace}, i.e.~$\myv g_B(\psi)$.

The valuation is preserved: denote ${(\sigma_\mathbb{P})_\mathbb{B}} \eqqcolon \tau$. Notice that $f_B \circ \sigma$ is given by $a \mapsto (\sigma(a))^* =  \{\Gamma \; | \; \Gamma \Vdash_\bb^\varnothing \sigma(a) \}$. Further, by Definitions~\ref{def:phasetobase} and~\ref{def:basetophase}, $\tau$ is given by $a \mapsto  \{\Gamma \; | \; \Gamma \Vdash_\bb^\varnothing \sigma(a) \}$. Hence, $f_B \circ \sigma = \tau$, implying  $f_B \circ \sigma(a) \dashv \vdash_{\bb} \tau(a)$. On the other hand, $g_B \circ \tau$ is given by $a \mapsto \{\Gamma \; | \; \Gamma \Vdash_\bb^\varnothing \sigma(a) \} \mapsto p$ such that $p^* = \{\Gamma \; | \; \Gamma \Vdash_\bb^\varnothing \sigma(a) \}$. Since $(\sigma(a))^* = \{\Gamma \; | \; \Gamma \Vdash_\bb^\varnothing \sigma(a) \} = p^*$, we have $p \dashv \vdash_{\bb} \sigma(a)$. Hence, $g_B \circ \tau(a) \dashv \vdash_{\bb} \sigma(a)$.

Finally, we show that the composition of base morphisms $f_B$ and $g_B$ yields an identity up to intersupport.

\begin{theorem}
    For every $\bb \in \mathbb{B}$, every $\phi \in (\forms_{\At})^\sigma$ and every $\psi \in (\forms_{\At})^{{(\sigma_\mathbb{P})_\mathbb{B}}}$ we have $g_B \circ f_B(\phi) \dashv \vdash_{\bb} \phi$ and $f_B \circ g_B(\psi) \dashv \vdash_{(\bb_\mathbb{P})_\mathbb{B}} \psi$.
\end{theorem}

\begin{proof}
    Let $\phi \in (\forms_{\At})^\sigma$. Then $g_B \circ f_B(\phi) = \psi$ such that $\psi^* = \brl f_B(\phi) \brr = \phi^*$ by Lemma~\ref{lmm:starissquarebracket2}. That is, $\{\Gamma \; | \; \Gamma \Vdash_\bb^\varnothing \phi\} = \{\Gamma \; | \; \Gamma \Vdash_\bb^\varnothing \psi\}$. Thus, $\psi \dashv \vdash_\bb \phi$, yielding $\psi = g_B \circ f_B(\phi) \dashv \vdash_{\bb} \phi$.

    Now, let $\psi \in (\forms_{\At})^{{(\sigma_\mathbb{P})_\mathbb{B}}}$. Then $f_B \circ g_B(\psi) = f_B(\phi)$ s.t. $\phi^* = \brl \psi \brr$. Since $\brl f_B(\phi) \brr = \phi^*$ by Lemma~\ref{lmm:starissquarebracket2}, we have that $f_B(\phi) \dashv \vdash_{(\bb_\mathbb{P})_\mathbb{B}} \phi^*$ by Lemma~\ref{lmm:intersupport}; hence, $f_B(\phi) \dashv \vdash_{(\bb_\mathbb{P})_\mathbb{B}} \brl \psi \brr$. Since $\brl \psi \brr \dashv \vdash_{(\bb_\mathbb{P})_\mathbb{B}} \psi $ by Lemma~\ref{lmm:intersupport}, we have that $f_B \circ g_B(\psi) \dashv \vdash_{(\bb_\mathbb{P})_\mathbb{B}} \psi$.
\end{proof}

\section{Concluding Remarks} \label{sec:conclusion}
The main contribution of this work is establishing an equivalence between a well-known algebraic semantics -- phase semantics -- and a recently developed base-extension semantics for linear logic. As a novel contribution to the base-extension semantics of classical linear logic, we have introduced the semantic clauses for the exponentials $(\bang)$ and $(\quest)$. 

Interestingly enough, the definitions of base and phase model morphisms that allow us to prove the equivalence are very similar in spirit to the completeness proofs (w.r.t. LL) of both phase and base-extension semantics. In the case of the phase morphism $f_P$ in Definition~\ref{def:monoidmap}, a so-called canonical phase model is constructed, but relative to the support in a given base, as opposed to the provability relation of LL in the completeness proof of phase semantics~\cite{DBLP:journals/tcs/Girard87, okada1998introduction}. Similarly, in the case of the phase morphism $f_B$ in Definition~\ref{def:basemap}, a so-called simulation base~\cite{DBLP:journals/igpl/Sandqvist15} is constructed, but one that mimics the derivability relation of the underlying phase model (i.e.~that between facts), as opposed to the derivability of LL~\cite{DBLP:journals/corr/abs-2504-08349}.

A similar construction involving quantifying over all multisets of formulas proving another formula is used in defining canonical models, e.g.~for the logic programming language based on $\top,\with,\multimap$ and $\Rightarrow$ in~\cite{DBLP:journals/iandc/HodasM94} and for the minimal non-normal modal logic in~\cite{Porello03072015}.

The~\ref{eq:supp-bang} clause presented in this paper differs from the intuitionistic $(\bang)$ clause in~\cite{DBLP:journals/corr/abs-2402-01982}, even if translated to the classical setting as proposed in the conclusion of~\cite{DBLP:journals/corr/abs-2504-08349}. The difference lies in the atomic multisets indexing the support relation:~\cite {DBLP:journals/corr/abs-2402-01982} has $\varnothing$ in place of $S$ in Definition~\ref{def:support}. Therefore, the two clauses only coincide when considering a base whose set of structural atoms is empty. Yet completeness proofs, which utilise a base containing structural atoms, are presented for both semantics. The same, and even greater, dissimilarity occurs between the $(\quest)$ clause conjectured in the conclusion of~\cite{DBLP:journals/corr/abs-2504-08349} and the~\ref{eq:supp-quest} clause given in this work: the difference lies not only in the empty multiset in place of a structural one but also in an extra multiset present in the support of $\bot$. The soundness proofs differ as well: the current paper makes use of structural atoms in the proof of Theorem~\ref{thm:soundness}, as one can weaken and contract on them, while no such atoms seem to appear in the corresponding proof in~\cite{DBLP:journals/corr/abs-2402-01982}. As of now, the difference between these approaches remains subtle and requires further investigation. 

Morphisms of bases and phase models make both into categories. It is therefore natural to ask whether the translations presented in Section~\ref{sec:corr} extend to functors by defining them on morphisms. Likewise, it seems reasonable to expect that the functions $f_P,g_P$ and $f_B,g_B$ from Definitions~\ref{def:monoidmap} and~\ref{def:basemap} give natural isomorphisms, hence exhibiting an equivalence of categories of bases and phase models. A subtlety arises for the functions $f_B,g_B$, as the isomorphism is not strict, but rather holds up to intersupport, suggesting higher categorical structure. This observation opens up the possibility of connecting this work to the category-theoretic interpretation of base-extension semantics~\cite{DBLP:journals/sLogica/PymRR25}.

Another prospect for future work is to investigate how results native to one semantics can be transferred to the other, and vice versa. Similarly, it is worth exploring whether the established equivalence allows one to prove the syntactic properties about linear logic without actually ever going into its syntax.

Finally, it would be interesting to explore the possible equivalence between phase semantics~\cite{kanovich2006intuitionistic,abrusci1990sequent} and B-eS~\cite{DBLP:conf/tableaux/GheorghiuGP23,DBLP:journals/corr/abs-2402-01982} for intuitionistic linear logic. The semantics differ from the classical ones as follows: in phase semantics, intuitionistic phase spaces with a closure relation other than double orthogonality are used, while in B-eS, the clauses quantify over all atoms in place of $\bot$. Given the relation between classical and intuitionistic phase spaces presented in~\cite{kanovich2006intuitionistic} and the similarity of B-eS clauses for classical and intuitionistic linear logic, it is reasonable to expect an intuitionistic counterpart of this work's result to hold as well.


\bibliography{bibliography} 

\appendix

\section{Exponentials in B-eS: soundness and completeness} \label{app:a}
Note that the content of this appendix builds upon the work in~\cite{DBLP:journals/corr/abs-2504-08349}. In particular, Lemma~\ref{lmm:raa}, Theorem~\ref{thm:soundness}, Lemma~\ref{lemma:completeness} and Theorem~\ref{thm:completeness} are extensions of~\cite[Lemma 4.1]{DBLP:journals/corr/abs-2504-08349},~\cite[Theorem 4.2]{DBLP:journals/corr/abs-2504-08349},~\cite[Lemma 5.3]{DBLP:journals/corr/abs-2504-08349} and~\cite[Theorem 5.4]{DBLP:journals/corr/abs-2504-08349}, accordingly, with inductive cases for $(\bang)$ and $(\quest)$.

\begin{lemma}[Monotonicity~\cite{DBLP:journals/corr/abs-2504-08349}] \label{lemma:monotonicity}
    If $\Gamma \Vdash^{L}_{\mathcal{B}} \phi$ and $\mathcal{C} \supseteq \mathcal{B}$, then $\Gamma \Vdash^{L}_{\mathcal{C}} \phi$.
\end{lemma}

\begin{lemma} \label{lemma:notbang}
    $\Vdash_\bb^L \bang \phi^\bot$ if and only if, for all $\bc \supseteq \bb$ and all $S$ structural in $\bb$, $\Vdash_\bc^S \phi$ implies $\Vdash_\bc^{S,L} \bot$.
\end{lemma}

\begin{proof}
    $(\Rightarrow):$ Assume that $\Vdash_\bb^L \bang \phi^\bot$ and $\Vdash_\bc^S \phi$ for an arbitrary structural $S$. Conclude $\Vdash_\bc^S \bang \phi$ by Lemma~\ref{lemma:phitobangphi}. Since also $\Vdash_\bb^L \bang \phi^\bot$ (hence $\bang \phi \Vdash_\bb^L \bot$), obtain $\Vdash_\bc^{S,L} \bot$ by~\ref{eq:supp-inf}, as required.

    $(\Leftarrow):$ Assume that $\Vdash_\bc^S \phi$ implies $\Vdash_\bc^{S,L} \bot$ for all structural $S$ and all $\bc \supseteq \bb$. Let $K$ be arbitrary and $\mathcal{X} \supseteq \bb$ be such that $\Vdash_{\mathcal{X}}^K \bang \phi$. Since $\mathcal{X} \supseteq \bb$, we have that $\Vdash_\mathcal{X}^S \phi$ implies $\Vdash_\mathcal{X}^{S,L} \bot$ (take $\bc = \mathcal{X}$). Hence, by~\ref{eq:supp-bang}, obtain $\Vdash_\mathcal{X}^{K,L} \bot$. Since $\mathcal{X} \supseteq \bb$ was such that $\Vdash_{\mathcal{X}}^K \bang \phi$, obtain $\bang \phi \Vdash_\bb^L \bot$ by~\ref{eq:supp-inf}, as desired.
\end{proof}

\begin{corollary} \label{cor:notbang}
    $\bang \psi_1,\dots, \bang \psi_n \Vdash_\bb^L \bang \phi^\bot$ if and only if, for all $\bc \supseteq \bb$ and all $S$ structural in $\bb$, $\Vdash_\bc^S \phi$ implies $\bang \psi_1,\dots, \bang \psi_n \Vdash_\bc^{S,L} \bot$.
\end{corollary}

\begin{lemma} \label{lmm:raa}
    If $\phi^\bot \Vdash^L_\bb \bot$ then $\Vdash^L_\bb \phi$.
\end{lemma}

\begin{proof}
Inductive cases for MALL can be found in~\cite{DBLP:journals/corr/abs-2504-08349}.
    \begin{description}
        \item[$\phi = \bang \alpha:$] Assume $(\bang \alpha)^\bot \Vdash^L_\bb \bot$. Now, assume, for all $\bd \supseteq \bc \supseteq \bb$, all $S$ structural in $\bb$ and an arbitrary $K$ that $\Vdash^S_\bd \alpha$ implies $\Vdash_\bd^{S,K} \bot$. Further assume, for an arbitrary $\mathcal{X} \supseteq \bc$ and an arbitrary $M$ that $\Vdash^M_{\mathcal{X}} \bang \alpha$. Then $\Vdash^{K,M}_\mathcal{X} \bot$ by~\ref{eq:supp-bang}. Since $\mathcal{X} \supseteq \bc$  such that $\Vdash^M_{\mathcal{X}} \bang \alpha$, obtain $\bang \alpha \Vdash^K_\bc \bot$ by~\ref{eq:supp-inf}. This is equivalent to $\Vdash^K_\bc (\bang \alpha)^\bot$, and since $(\bang \alpha)^\bot \Vdash^L_\bb \bot$, obtain $\Vdash_\bc^{L,K} \bot$ by~\ref{eq:supp-inf}. Finally, since $\bd \supseteq \bc \supseteq \bb$ was such that $\Vdash^S_\bd \alpha$ implies $\Vdash_\bd^{S,K} \bot$, we have that $\Vdash^L_\bb \bang \alpha$ by~\ref{eq:supp-bang}.

        \item[$\phi = \quest \alpha:$] Assume $(\quest \alpha)^\bot \Vdash^L_\bb \bot$. Now, assume, for an arbitrary $\bc \supseteq \bb$ and an arbitrary $S$ structural in $\bb$ that $\alpha \Vdash^S_\bc \bot$. Further assume, for an arbitrary $\bd \supseteq \bc$ and $K$ that $\Vdash^K_\bd \quest \alpha$. Then $\Vdash^{S,K}_\bd \bot$ by~\ref{eq:supp-quest}. Since $\bd \supseteq \bc$ such that $\Vdash^K_\bd \quest \alpha$, obtain $\quest \alpha \Vdash^S_\bc \bot$ by~\ref{eq:supp-inf}. This is equivalent to $\Vdash^S_\bc (\quest \alpha)^\bot$, and since $(\quest \alpha)^\bot \Vdash^L_\bb \bot$, obtain $\Vdash_\bc^{L,S} \bot$ by~\ref{eq:supp-inf}. Finally, since $\bc \supseteq \bb$ such that $\alpha \Vdash^S_\bc \bot$, we have that $\Vdash^L_\bb \quest \alpha$ by~\ref{eq:supp-quest}.
    \end{description}
\end{proof}

Note that the following Lemmas~\ref{lemma:generictensor}--\ref{lemma:genericzero}, with inductive cases for MALL covered in~\cite{DBLP:journals/corr/abs-2504-08349}, hold for $\chi = \bang \alpha$ and $\chi = \quest \alpha$. In fact, we give a shorter proof for Lemma~\ref{lemma:generictensor} (by making use of Lemma~\ref{lmm:raa}) that can be easily adapted for the rest.

\begin{lemma} \label{lemma:generictensor}
    If $\Vdash^{L}_{\mathcal{B}} \phi \otimes \psi$ and $\phi, \psi \Vdash^{K}_{\mathcal{B}} \chi$ then $\Vdash^{L, K}_{\mathcal{B}} \chi$.
\end{lemma}

\begin{proof}
    Assume, for an arbitrary $\bc \supseteq \bb$ and an arbitrary $M$ that $\Vdash_\bc^M \chi^\bot$. This is equivalent to $\chi \Vdash_\bc^M \bot$. Further assume, for an arbitrary $\bd \supseteq \bc$ and arbitrary $N,P$, that $\Vdash^N_\bd \phi$ and  $\Vdash^P_\bd \psi$. Since $\phi, \psi \Vdash^{K}_{\bb} \chi$, obtain $\Vdash^{K,N,P}_{\bd} \chi$ by~\ref{eq:supp-inf}. Now, since $\chi \Vdash_\bc^M \bot$, obtain $\Vdash^{K,N,P,M}_{\bd} \bot$ by~\ref{eq:supp-inf} again. Since $\bd \supseteq \bc$ such that $\Vdash^N_\bd \phi$ and  $\Vdash^P_\bd \psi$, we have that $\phi,\psi \Vdash_\bc^{K,M} \bot$ by~\ref{eq:supp-inf}. Now, since $\Vdash^{L}_{\mathcal{B}} \phi \otimes \psi$, obtain $\Vdash_\bc^{L,K,M} \bot$ by~\ref{eq:supp-tensor}. Finally, since $\bc \supseteq \bb$ such that $\Vdash_\bc^M \chi^\bot$, obtain $\chi^\bot \Vdash_\bb^{L,K} \bot$ by~\ref{eq:supp-inf}. Hence, $\Vdash_\bb^{L,K} \chi$ by Lemma~\ref{lmm:raa}.
\end{proof}

\begin{lemma} \label{lemma:genericimplication}
    If $\Vdash^{L}_{\mathcal{B}} \phi^\psi$ and $\Vdash^{K}_{\mathcal{B}} \phi$ and $\psi \Vdash^{M}_{\mathcal{B}} \chi$ then $\Vdash^{L, K, M}_{\mathcal{B}} \chi$.
\end{lemma}

\begin{lemma} \label{lemma:genericone}
    If $\Vdash^{L}_{\mathcal{B}} 1$ and $\Vdash^{K}_{\mathcal{B}} \chi$ then $\Vdash^{L, K}_{\mathcal{B}} \chi$.
\end{lemma}

\begin{lemma} \label{lemma:genericplus}
    If $\Vdash^{L}_{\mathcal{B}} \phi \oplus \psi$ and $\phi \Vdash^{K}_{\mathcal{B}} \chi$ and $\psi \Vdash^{K}_{\mathcal{B}} \chi$ then $\Vdash^{L, K}_{\mathcal{B}} \chi$.
\end{lemma}

\begin{lemma} \label{lemma:genericzero}
    If $\Vdash^{L}_{\mathcal{B}} 0$ then $\Vdash^{L, K}_{\mathcal{B}} \chi$.
\end{lemma}

\begin{theorem}[Soundness] \label{thm:soundness}
    If $\Gamma \vdash_{LL} \phi$ then $\Gamma \Vdash \phi$.
\end{theorem}

\begin{proof}
    Given that $\vdash$ is defined inductively, it suffices to prove the following:
    \begin{description}
        \item[($\bang$I)] If $\Gamma_1 \Vdash \bang\psi_1, \dots, \Gamma_n \Vdash \bang\psi_n$ and $\bang \psi_1,\dots,\bang \psi_n \Vdash \phi$ then $\Gamma_1,\dots,\Gamma_n \Vdash \bang\phi$.
        \item[($\bang$E)] If $\Gamma \Vdash \bang \phi$ then $\Gamma \Vdash \phi$.
        \item[(Wk)'] If $\Gamma \Vdash \bang \phi$ and $\Delta \Vdash \psi$ then $\Gamma,\Delta \Vdash \psi$.
        \item[(Ctr)'] If $\Gamma \Vdash \bang \phi$ and $\Delta, \bang \phi, \bang \phi \Vdash \psi$ then $\Gamma,\Delta \Vdash \psi$.
        \item[($\quest$I)] If $\Gamma \Vdash \phi$ then $\Gamma \Vdash \quest \phi$.
        \item[($\quest$E)] If $\Gamma \Vdash \quest \phi$ and $\Delta_1 \Vdash \bang\psi_1, \dots, \Delta_n \Vdash \bang\psi_n$ and $\phi,\bang \psi_1,\dots,\bang \psi_n \Vdash \chi$ then $\Gamma, \Delta_1,\dots,\Delta_n \Vdash \chi$, provided that $\chi$ is of the form $\quest \delta$ or $\bot$.
    \end{description}

    Inductive cases for MALL, i.e.~statements corresponding to the natural deduction rules of MALL $(\otimes I)'$--$(0E)'$, are proven with the help of Lemmas~\ref{lemma:generictensor}--\ref{lemma:genericzero} and can be found in~\cite[Theorem 4.2]{DBLP:journals/corr/abs-2504-08349}.

    We shall abuse notation and write $\Vdash_{\bb}^L \Gamma$ for $\Vdash_{\bb}^{L_1} \gamma_1,\dots,\Vdash_{\bb}^{L_n} \gamma_n$, where $\Gamma = \{\gamma_1,\dots,\gamma_n \}$ and $L = L_1 \cup\dots \cup L_n$. Now, assume, for an arbitrary $\bb$ and an arbitrary $L$, that $\Vdash_{\bb}^L \Gamma$.

    \begin{description}
        \item[($\bang$I)'.] Assume, for all $\bd \supseteq \bc \supseteq \bb$, all $S$ structural in $\bb$ and an arbitrary $M$, that $\Vdash_\bd^S \phi$ implies $\Vdash_\bd^{S,M} \bot$. Further assume, for an arbitrary $\mathcal{X} \supseteq \bb$ and structural $S_1,\dots,S_n$, that $\Vdash_\mathcal{X}^{S_1} \bang \psi_1,\dots,\Vdash_\mathcal{X}^{S_n} \bang \psi_n$. Since also $\bang \psi_1,\dots,\bang \psi_n \Vdash_\bb^\varnothing \phi$, obtain $\Vdash_\mathcal{X}^{S_1,\dots,S_n} \phi$ by~\ref{eq:supp-inf}. Since $\mathcal{X}$ is its own extension and $\Vdash_\mathcal{X}^{S_1,\dots,S_n} \phi$, obtain $\Vdash_{\mathcal{X}}^{S_1,\dots,S_n,M} \bot$ by our original assumption (take $\bc = \mathcal{X}$). Hence obtain $\bang \psi_1,\dots,\bang \psi_n \Vdash_\bb^M \bot$ by Corollary~\ref{cor:notbang}. Since also $\Vdash_\bb^{L_1} \bang \psi_1,\dots,\Vdash_\bb^{L_n} \bang \psi_n$, obtain $\Vdash_\bb^{M, L_1,\dots,L_n} \bot$ by~\ref{eq:supp-inf}. Since also $\bd \supseteq \bc \supseteq \bb$ was such that $\Vdash_\bd^S \phi$ implies $\Vdash_\bd^{S,M} \bot$, obtain $\Vdash_\bb^{L_1,\dots,L_n} \bang \phi$ by~\ref{eq:supp-bang}, as required.
        
        \item[($\bang$E)'.] Assume, for an arbitrary $\bc \supseteq \bb$ and $K$, that $\Vdash_\bc^K \phi^\bot$. Then, for all $\bd \supseteq \bc$ and all structural $S$, $\Vdash_\bd^S \phi$ implies $\Vdash_\bd^{S,K} \bot$ by~\ref{eq:supp-inf}. Since also $\Vdash_\bb^L \bang \phi$, obtain $\Vdash_\bc^{L,K} \bot$ by~\ref{eq:supp-bang}. Since $\bc \supseteq \bb$ such that $\Vdash_\bc^K \phi^\bot$, obtain $\phi^\bot \Vdash_\bb^L \bot$ by~\ref{eq:supp-inf}. Then, by Lemma~\ref{lmm:raa}, $\Vdash_\bb^L \phi$ as expected.
        
        \item[(Weakening)'.] Assume, for an arbitrary $\bc \supseteq \bb$ and $K$, that $\Vdash_\bc^K \psi^\bot$ (which is equivalent to $\psi \Vdash_\bc^K \bot$). Since also $\Vdash_\bb^M \psi$, obtain $\Vdash_\bc^{M,K} \bot$ by~\ref{eq:supp-inf}. Further assume, for an arbitrary $\bd \supseteq \bc$ and structural $S$, that $\Vdash_\bd^S \phi$. By monotonicity, also $\Vdash_\bd^{M,K} \bot$ and, by structurality of S, $\Vdash_\bd^{S,M,K} \bot$. Since $\Vdash_\bd^S \phi$ implied $\Vdash_\bd^{S,M,K} \bot$, and also $\Vdash_\bb^L \bang \phi$, obtain $\Vdash_\bc^{L,M,K} \bot$ by~\ref{eq:supp-bang}. Since $\bc \supseteq \bb$ such that $\Vdash_\bc^K \psi^\bot$, obtain $\psi^\bot \Vdash_\bb^{L,M} \bot$ by~\ref{eq:supp-inf}. Then, by Lemma~\ref{lmm:raa}, $\Vdash_\bb^{L,M} \psi$ as expected.
        
        \item[(Contraction)'.] Assume, for an arbitrary $\bc \supseteq \bb$ and $K$, that $\Vdash_\bc^K \psi^\bot$. Since also $\bang \phi, \bang \phi \Vdash_\bb^M \psi$, obtain $\bang \phi, \bang \phi \Vdash_\bc^{M,K} \bot$ by~\ref{eq:supp-inf}. Since $\Vdash_\bb^L \bang \phi$, obtain $\bang \phi \Vdash_\bc^{L,M,K} \bot$ by~\ref{eq:supp-inf}. Assume, for an arbitrary $\bd \supseteq \bc$ and structural $S$, that $\Vdash_\bd^S \phi$. Then, by Lemma~\ref{lemma:notbang}, $\Vdash_\bd^{S,L,M,K} \bot$. Then, by~\ref{eq:supp-inf}, $\bang \phi \Vdash_\bd^{S,M,K} \bot$. By Lemma~\ref{lemma:notbang} again, $\Vdash_\bd^{S,S,M,K} \bot$; hence $\Vdash_\bd^{S,M,K} \bot$ by structurality of $S$. Since $\bd \supseteq \bc$ such that $\Vdash_\bd^S \phi$, obtain $\bang \phi \Vdash_\bc^{M,K} \bot$ by Lemma~\ref{lemma:notbang}. Since also $\Vdash_\bb^L \bang \phi$, obtain $\Vdash_\bc^{L,M,K} \bot$ by~\ref{eq:supp-inf}. Since $\bc \supseteq \bb$ such that $\Vdash_\bc^K \psi^\bot$, obtain $\psi^\bot \Vdash_\bb^{L,M} \bot$ by~\ref{eq:supp-inf}. Then, by RAA, $\Vdash_\bb^{L,M} \psi$ as expected.
        
        \item[($\quest$I)'.] Assume, for an arbitrary $\bc \supseteq \bb$ and a structural $S$, that $\phi \Vdash_\bc^S \bot$. Since also $\Vdash_\bb^L \phi$, obtain $\Vdash_\bc^{L,S} \bot$ by~\ref{eq:supp-inf}. Since $\bc \supseteq \bb$ was such that $\phi \Vdash_\bc^S \bot$, obtain $\Vdash_\bb^L \quest \phi$ by~\ref{eq:supp-quest}, as desired.
        
        \item[($\quest$E)'.] $\chi = \bot$: Assume, for an arbitrary $\bc \supseteq \bb$ and structural $S_1,\dots,S_n$, that $\Vdash_\bc^{S_1} \bang \psi_1,\dots,\Vdash_\bc^{S_n} \bang \psi_n$. Since also $\phi, \bang \psi_1,\dots,\bang \psi_n \Vdash_\bb^\varnothing \chi$, obtain $\phi \Vdash_\bc^{S_1,\dots,S_n} \chi$ by~\ref{eq:supp-inf}. Since also $\Vdash_\bb^L \quest \phi$, obtain $\Vdash_\bc^{L, S_1,\dots,S_n} \bot$ by~\ref{eq:supp-quest}. Hence obtain $\bang \psi_1,\dots,\bang \psi_n \Vdash_\bb^L \bot$ by Corollary~\ref{cor:notbang}. Since also $\Vdash_\bb^{K_1} \bang \psi_1,\dots,\Vdash_\bb^{K_n} \bang \psi_n$, obtain $\Vdash_\bb^{L, K_1,\dots,K_n} \bot$ by~\ref{eq:supp-inf}, as required. 
        
        $\chi = \quest \delta$: Assume, for an arbitrary $\bc \supseteq \bb$ and structural $T$, that $\delta \Vdash_\bc^T \bot$. Further assume, for an arbitrary $\bd \supseteq \bc$, $M$ and structural $S_1,\dots,S_n$, that $\Vdash_\bd^{S_1} \bang \psi_1,\dots,\Vdash_\bd^{S_n} \bang \psi_n$ and $\Vdash^M_\bd \phi$. Since also $\phi, \bang \psi_1,\dots,\bang \psi_n \Vdash_\bb^\varnothing \chi$, obtain $\Vdash_\bd^{S_1,\dots,S_n, M} \quest \delta$ by~\ref{eq:supp-inf}. Since $\delta \Vdash_\bc^T \bot$, obtain $\Vdash_\bd^{T, S_1,\dots,S_n, M} \bot$ by~\ref{eq:supp-quest}. By~\ref{eq:supp-inf} again, obtain $\phi \Vdash_\bd^{T, S_1,\dots,S_n} \bot$.  Since also $\Vdash_\bb^L \quest \phi$, obtain $\Vdash_\bd^{L, T, S_1,\dots,S_n} \bot$ by~\ref{eq:supp-quest}. Hence obtain $\bang \psi_1,\dots,\bang \psi_n \Vdash_\bc^{L,T} \bot$ by Corollary~\ref{cor:notbang}. Since also $\Vdash_\bb^{K_1} \bang \psi_1,\dots,\Vdash_\bb^{K_n} \bang \psi_n$, obtain $\Vdash_\bc^{L, T, K_1,\dots,K_n} \bot$ by~\ref{eq:supp-inf}. Since $\bc \supseteq \bb$ was such that $\delta \Vdash_\bc^T \bot$, obtain $\Vdash_\bb^{L, K_1,\dots,K_n} \quest \delta$ by~\ref{eq:supp-quest}, as required.
    \end{description}
\end{proof}

\begin{definition}[Simulation base]
Let $\mathcal{U}$ be a \emph{simulation base} as in~\cite[Definition 5.2]{DBLP:journals/corr/abs-2504-08349} with the addition of the following rules, where $s_i$ are structural, $\delta \in \{\quest \alpha, \bot \}$ for some $\alpha$ and $\chi \in \{\bang \beta, \neg \quest \beta \}$ for some $\beta$.
    \begin{center}
    \noindent\begin{minipage}{0.51\textwidth}\scriptsize
        \begin{prooftree}
        \def\ScoreOverhang{0.3pt}
            \AxiomC{$L_1 \vdash s_1$}
            \AxiomC{$\ldots$}
            \AxiomC{$L_n \vdash s_n$}
            \AxiomC{$\varnothing;s_1,\ldots, s_n \vdash p^{\phi}$}
            \RightLabel{$\bang$I}
            \QuaternaryInfC{$L_1, \ldots,L_n \vdash p^{\bang\phi}$}
        \end{prooftree}
        \begin{prooftree}
        \def\ScoreOverhang{0.3pt}
            \AxiomC{$K \vdash p^{\quest\phi}$}
            \AxiomC{$L_1 \vdash s_1 \; \ldots \; L_n \vdash s_n$}
            \AxiomC{$\varnothing; s_1,\ldots,s_n, p^\phi \vdash p^{\delta}$}
            \RightLabel{$\quest$E}
            \TrinaryInfC{$K,L_1, \ldots,L_n \vdash p^{\delta}$}
        \end{prooftree}
    \end{minipage}%
    \noindent\begin{minipage}{0.22\textwidth}\scriptsize
        \begin{prooftree}
        \def\ScoreOverhang{0.5pt}
            \AxiomC{$L \vdash p^{\bang \phi}$}
            \RightLabel{$\bang$E}
            \UnaryInfC{$L \vdash p^{\phi}$}
        \end{prooftree}
        \begin{prooftree}
        \def\ScoreOverhang{0.5pt}
            \AxiomC{$L \vdash p^{\phi}$}
            \RightLabel{$\quest$I}
            \UnaryInfC{$L \vdash p^{\quest \phi}$}
        \end{prooftree}
    \end{minipage}
    \noindent\begin{minipage}{0.22\textwidth}\scriptsize
        \begin{prooftree}
        \def\ScoreOverhang{0.5pt}
            \AxiomC{$L \vdash p^{\chi}$}
            \AxiomC{$K \vdash p^\psi$}
            \RightLabel{Wk}
            \BinaryInfC{$L, K \vdash p^\psi$}
        \end{prooftree}
        \begin{prooftree}
        \def\ScoreOverhang{0.5pt}
            \AxiomC{$L \vdash p^{\chi}$}
            \AxiomC{$K;p^{\chi},p^{\chi} \vdash p^\psi$}
            \RightLabel{Ctr}
            \BinaryInfC{$L, K \vdash p^\psi$}
        \end{prooftree}
    \end{minipage}
    \noindent\begin{minipage}{\textwidth}\scriptsize
        \begin{prooftree}
        \def\ScoreOverhang{0.5pt}
            \AxiomC{$\Gamma_{\At}, p^{\neg \phi} \vdash \bot$}
            \RightLabel{Raa}
            \UnaryInfC{$\Gamma_{\At} \vdash p^{\phi}$}
        \end{prooftree}
    \end{minipage}
    \end{center}
\end{definition}

\begin{lemma} \label{lemma:struct}
    Let $\mathcal{U}$ be a simulation base and $\phi$ a formula. An atom in $\mathcal{U}$ is structural if and only if it is of the shape $p^{\bang \phi}, p^{\neg \quest \phi}$.
\end{lemma}

\begin{proof}
    Recall that $\mathcal{U}$ contains the following rules, where $\chi \in \{\bang \phi, \neg \quest \phi\}$ for some $\phi$:

    \begin{center}
        \begin{bprooftree}
        \def\ScoreOverhang{0.5pt}
            \AxiomC{$L \vdash p^{\chi}$}
            \AxiomC{$K \vdash p^\beta$}
            \RightLabel{\small Wk}
            \BinaryInfC{$L, K \vdash p^\beta$}
        \end{bprooftree}
        \quad
        \begin{bprooftree}
        \def\ScoreOverhang{0.5pt}
            \AxiomC{$L \vdash p^{\chi}$}
            \AxiomC{$K;p^{\chi},p^{\chi} \vdash p^\beta$}
            \RightLabel{\small Ctr}
            \BinaryInfC{$L, K \vdash p^\beta$}
        \end{bprooftree}
    \end{center}
    
    \noindent$(\Rightarrow)$: Consider a structural atom $q$ in $\mathcal{U}$. By Definition~\ref{def:structuralatoms}, $\mathcal{U}$ must contain the rules of weakening and contraction for $q$. The only rules of such shape in $\mathcal{U}$ are the rules above with $p^\beta = \bot$, hence $q$ must be of the shape $p^{\bang \phi}, p^{\neg \quest \phi}$.
    
    \noindent$(\Leftarrow)$: Consider $p^{\psi}$, where $\psi \in \{\bang \phi, \neg \quest \phi\}$ for some $\phi$. Now, by Definition~\ref{def:structuralatoms} and setting $p^\beta = \bot$ in Wk and Ctr, we conclude that $p^{\psi}$ is structural.
\end{proof}

\begin{lemma}[Completeness lemma] \label{lemma:completeness}
    Let $\mathcal{U}$ be a simulation base and $\phi$ a formula. Then, for all $\mathcal{B} \supseteq \mathcal{U}$ and all $L$, $\Vdash^{L}_{\mathcal{B}} \phi$ if and only if $L \vdash_{\mathcal{B}} p^{\phi}$.
\end{lemma}

\begin{proof}
    We prove the result by induction on the complexity of $\phi$. The induction hypothesis is such that $$\Vdash^{K}_{\mathcal{B}} \chi \text{ if and only if } K \vdash_{\mathcal{B}} p^{\chi}$$ holds true for any $K$, any subformula $\chi$ of $\phi$ and any base in place of $\mathcal{B}$. Inductive cases for MALL can be found in~\cite{DBLP:journals/corr/abs-2504-08349}.

    \begin{description}
        \item[($\bang$).] $\phi = \bang \alpha$.

        \noindent$(\Rightarrow)$: Assume $\Vdash^{L}_{\mathcal{B}}  \bang \alpha$. Further assume, for an arbitrary $\bc \supseteq \bb$ and an arbitrary structural $S$, that $\Vdash_\bc^S \alpha$. Induction hypothesis yields $S \vdash_\bc p^\alpha$. Now, the following is a deduction in $\mathcal{C}$, where $S = \{s_1,\dots,s_n\}$:
        
        \begin{center}
            \begin{bprooftree}
            \def\ScoreOverhang{0.5pt}
                \AxiomC{}
                \UnaryInfC{$s_1 \vdash s_1$}
                \AxiomC{$\dots$}
                \AxiomC{}
                \UnaryInfC{$s_n \vdash s_n$}
                \AxiomC{$\varnothing; s_1,\dots,s_n \vdash p^\alpha$}
                \RightLabel{\small $\bang\text{I}$}
                \QuaternaryInfC{$s_1,\dots,s_n \vdash p^{\bang \alpha}$}
                \AxiomC{}
                \RightLabel{\small Ax}
                \UnaryInfC{$p^{\neg \bang \alpha} \vdash p^{\neg \bang \alpha}$}
                \RightLabel{\small $\multimap$E}
                \BinaryInfC{$s_1,\dots,s_n, p^{\neg \bang \alpha} \vdash \bot$}
            \end{bprooftree} 
        \end{center}
        
        This deduction shows $S, p^{\neg \bang \alpha} \vdash_\bc \bot$, hence $\Vdash_\bc^{S, p^{\neg \bang \alpha}} \bot$. Since $\bc \supseteq \bb$ was such that $\Vdash_\bc^S \alpha$, and also $\Vdash^{L}_{\mathcal{B}} \bang \alpha$, obtain $\Vdash_\bb^{L, p^{\neg \bang \alpha}} \bot$ by~\ref{eq:supp-bang}. Then, $L, p^{\neg \bang \alpha} \vdash_\bb \bot$. By Raa, $L \vdash_\bb p^{\bang \alpha}$, as desired. \\
        
        \noindent$(\Leftarrow)$: Assume $L \vdash_\bb p^{\bang \alpha}$. Further assume, for an arbitrary $\bc \supseteq \bb$, all $\bd \supseteq \bc$, all structural $S$ and an arbitrary $K$, that $\Vdash_\bd^S \alpha$ implies $\Vdash_\bd^{S,K} \bot$. The following is a deduction in $\mathcal{B}$:
        
        \begin{center}
            \begin{bprooftree}
            \def\ScoreOverhang{0.5pt}
                \AxiomC{}
                \RightLabel{\small Ax}
                \UnaryInfC{$p^{\bang \alpha} \vdash p^{\bang \alpha}$}
                \RightLabel{\small $\bang\text{E}$}
                \UnaryInfC{$p^{\bang \alpha} \vdash p^{\alpha}$}
            \end{bprooftree} 
        \end{center}

        The deduction yields $p^{\bang \alpha} \vdash_\bb p^{\alpha}$, hence, by induction hypothesis, $\Vdash_\bb^{p^{\bang \alpha}} \alpha$. By monotonicity, $\Vdash_\bc^{p^{\bang \alpha}} \alpha$ and, since $p^{\bang \alpha}$ is structural (Lemma~\ref{lemma:struct}) and $\bc$ is its own extension, obtain $\Vdash_\bc^{K, p^{\bang \alpha}} \bot$ from our initial assumption. Hence, $K, p^{\bang \alpha} \vdash_\bc \bot$. Since also $L \vdash_\bb p^{\bang \alpha}$, we obtain $K, L \vdash_\bc \bot$ by composing the two deductions. Hence, $\Vdash_\bc^{K, L} \bot$. Now, since $\bc \supseteq \bb$ was arbitrary and $\bd \supseteq \bc$ was such that $\Vdash_\bd^S \alpha$ implies $\Vdash_\bd^{S,K} \bot$, obtain $\Vdash_\bb^L \bang \alpha$ by~\ref{eq:supp-bang}, as required.

        \item[($\quest$).] $\phi = \quest \alpha$.

        \noindent$(\Rightarrow)$: Assume $\Vdash^{L}_{\mathcal{B}} \quest \alpha$. Further assume, for an arbitrary $\bc \supseteq \bb$ and an arbitrary $K$, that $\Vdash_\bc^K \alpha$. Induction hypothesis yields $K \vdash_\bc p^\alpha$. Now, the following is a deduction in $\mathcal{C}$:

        \begin{center}
            \begin{bprooftree}
            \def\ScoreOverhang{0.5pt}
                \AxiomC{$K \vdash p^{\alpha}$}
                \RightLabel{\small $\quest\text{I}$}
                \UnaryInfC{$K \vdash p^{\quest \alpha}$}
                \AxiomC{}
                \RightLabel{\small Ax}
                \UnaryInfC{$p^{\neg \quest \alpha} \vdash p^{\neg \quest \alpha}$}
                \RightLabel{\small $\multimap$E}
                \BinaryInfC{$K, p^{\neg \quest \alpha} \vdash \bot$}
            \end{bprooftree} 
        \end{center}
        
        This deduction shows $K, p^{\neg \quest \alpha} \vdash_\bc \bot$, hence $\Vdash_\bc^{K, p^{\neg \quest \alpha}} \bot$. Since $\bc \supseteq \bb$ was such that $\Vdash_\bc^K \alpha$ for an arbitrary $K$, obtain $\alpha \Vdash_\bb^{p^{\neg \quest \alpha}} \bot$ by~\ref{eq:supp-inf}. Now, since $\Vdash^{L}_{\mathcal{B}} \quest \alpha$, and $p^{\neg \quest \alpha}$ is structural (Lemma~\ref{lemma:struct}), obtain $\Vdash_\bb^{L,p^{\neg \quest \alpha}} \bot$ by~\ref{eq:supp-quest}. Thus, $L, p^{\neg \quest \alpha} \vdash_\bb \bot$. By Raa, $L \vdash_\bb p^{\quest \alpha}$, as desired. \\
        
        \noindent$(\Leftarrow)$: Assume $L \vdash_\bb p^{\quest \alpha}$. Further assume, for an arbitrary $\bc \supseteq \bb$ and an arbitrary structural $S$, that $\alpha \Vdash_\bc^S \bot$. Notice that $p^\alpha \vdash_\bb p^\alpha$, hence, by induction hypothesis, $\Vdash_\bb^{p^\alpha} \alpha$. Since also $\alpha \Vdash_\bc^S \bot$, obtain $\Vdash_\bc^{S,p^\alpha} \bot$ by~\ref{eq:supp-inf}. Thus, $S,p^\alpha \vdash_\bc \bot$. Now, the following is a deduction in $\mathcal{C}$, where $S = \{s_1,\dots,s_n\}$:
        
        \begin{prooftree}
        \def\ScoreOverhang{0.5pt}
            \AxiomC{$L \vdash p^{\quest \alpha}$}
            \AxiomC{}
            \UnaryInfC{$s_1 \vdash s_1$}
            \AxiomC{$\dots$}
            \AxiomC{}
            \UnaryInfC{$s_n \vdash s_n$}
            \AxiomC{$\varnothing; s_1,\dots,s_n, p^\alpha \vdash \bot$}
            \RightLabel{\small $\quest$E}
            \QuinaryInfC{$L, s_1,\dots,s_n \vdash \bot$}
        \end{prooftree}
        
        This deduction shows $L, S \vdash_\bc \bot$, hence, $\Vdash_\bc^{L,S} \bot$. Now, since $\bc \supseteq \bb$ was such that $\alpha \Vdash_\bc^S \bot$ for an arbitrary $S$, obtain $\Vdash_\bb^L \quest \alpha$ by~\ref{eq:supp-quest}, as required.
    \end{description}
\end{proof}

\begin{theorem}[Completeness] \label{thm:completeness}
    If $\Gamma \Vdash \phi$ then $\Gamma \vdash_{LL} \phi$.
\end{theorem}

\begin{proof}
    The proof is exactly the same as in~\cite{DBLP:journals/corr/abs-2504-08349}, but with the extended version of the completeness lemma above.
\end{proof}

\section{Omitted proofs from Section~\ref{sec:corr}} \label{app:proofs}

\begin{proof}[Proof of Lemma~\ref{lmm:intersupport}]
    Denote $\bb \coloneqq P_\mathbb{B}$. We prove the statement by induction on the structure of $\phi$.
    \begin{description}
        \item[$\phi = p:$] Immediate by construction of $\bb$ since $\brl p \brr = p$.
        
        \item[$\phi = \alpha^\bot:$] Suppose for an arbitrary $\bc \supseteq \bb$ and $L$ that $\Vdash_\bc^L \alpha^\bot$. Then $\alpha \Vdash_\bc^L \bot$ and, by inductive hypothesis, $\brl \alpha \brr \Vdash_\bc^L \bot$. Since $\brl \alpha^{\bot\bot} \brr \subseteq \brl \alpha \brr$, $\brl \alpha^{\bot\bot} \brr \Vdash_\bc^L \bot$ and since \textit{reductio ad absurdum} is present for every atom in $\bb$, obtain $\Vdash_\bc^L \brl \alpha^{\bot} \brr$. Hence, by~\ref{eq:supp-inf}, $\alpha^\bot \Vdash_\bb^\varnothing \brl \alpha^\bot \brr$.

        Now, assume for an arbitrary $\bc \supseteq \bb$ and $L$ that $\Vdash_\bc^L \brl \alpha^\bot \brr$. Since $\brl \alpha \brr \otimes_M \brl \alpha^\bot \brr \subseteq \bot$, we have that $\brl \alpha \brr,\brl \alpha^\bot \brr \vdash_\bb \bot$, hence $\brl \alpha \brr,\brl \alpha^\bot \brr \Vdash_\bb^\varnothing \bot$. By inductive hypothesis, $\alpha,\brl \alpha^\bot \brr \Vdash_\bb^\varnothing \bot$, and since $\Vdash_\bc^L \brl \alpha^\bot \brr$, obtain $\alpha \Vdash_\bb^L \bot$,~i.e.~$\Vdash_\bb^L \alpha^\bot$. Hence $\brl \alpha^\bot \brr \Vdash_\bb^\varnothing \alpha^\bot$ by~\ref{eq:supp-inf}.
        
        \item[$\phi = \alpha \otimes \beta:$] Suppose for an arbitrary $\bc \supseteq \bb$ and $L$ that $\Vdash_\bc^L \alpha \otimes \beta$. Since $\alpha \Vdash_\bb^\varnothing \brl \alpha \brr$ and $\beta \Vdash_\bb^\varnothing \brl \beta \brr$ by inductive hypothesis, obtain $\alpha,\beta \Vdash_\bb^\varnothing \brl \alpha \brr \otimes \brl \beta \brr$ by~$(\otimes I)'$~\cite[Theorem 4.2]{DBLP:journals/corr/abs-2504-08349}. Since also $\Vdash_\bc^L \alpha \otimes \beta$, obtain $\Vdash_\bc^L \brl \alpha \brr \otimes \brl \beta \brr$ by Lemma~\ref{lemma:generictensor}. Since $\brl \alpha \brr \otimes_M \brl \beta \brr \subseteq \brl \alpha \otimes \beta \brr$, we have that $\brl \alpha \brr, \brl \beta \brr \Vdash_\bb^\varnothing \brl \alpha \otimes \beta \brr$. By Lemma~\ref{lemma:generictensor}, obtain $\Vdash_\bc^L \brl \alpha \otimes \beta \brr$. Hence, by~\ref{eq:supp-inf}, $\alpha \otimes \beta \Vdash_\bb^\varnothing \brl \alpha \otimes \beta \brr$.

        Now, assume for an arbitrary $\bc \supseteq \bb$ and $L$ that $\Vdash_\bc^L \brl \alpha \otimes \beta \brr$. Further, for an arbitrary $\bd \supseteq \bc$ and $K$, assume that $\alpha,\beta \Vdash_\bd^K \bot$. By inductive hypothesis, $\brl \alpha \brr,\brl \beta \brr \Vdash_\bd^K \bot$. Since $\brl \alpha \otimes \beta \brr \subseteq \brl \alpha \brr \otimes_M \brl \beta \brr$, we have that $\brl \alpha \otimes \beta \brr \Vdash_\bb^\varnothing \brl \alpha \brr \otimes \brl \beta \brr$. By~\ref{eq:supp-inf}, obtain $\Vdash_\bc^L \brl \alpha \brr \otimes \brl \beta \brr$. Since also $\brl \alpha \brr,\brl \beta \brr \Vdash_\bd^K \bot$, obtain $\Vdash_\bd^{K,L} \bot$ by~\ref{eq:supp-tensor}. Since $\bd \supseteq \bc$ was such that $\alpha,\beta \Vdash_\bd^K \bot$, obtain $\Vdash_\bc^L \alpha \otimes \beta$ by~\ref{eq:supp-tensor}. By~\ref{eq:supp-inf} again, $\brl \alpha \otimes \beta \brr \Vdash_\bb^\varnothing \alpha \otimes \beta$.
        
        \item[$\phi = 1:$] Suppose for an arbitrary $\bc \supseteq \bb$ and $L$ that $\Vdash_\bc^L 1$, which is equivalent to $\Vdash_\bc^L \bot^\bot$. By inductive hypothesis, $\Vdash_\bc^L \brl \bot^\bot \brr$. Now, by soundness of phase semantics $\brl \bot^\bot \brr \subseteq \brl 1 \brr$, hence $\brl \bot^\bot \brr \Vdash_\bb^\varnothing \brl 1 \brr$. By~\ref{eq:supp-inf}, $\Vdash_\bb^L \brl 1 \brr$, and by~\ref{eq:supp-inf} again, $1 \Vdash_\bb^\varnothing \brl 1 \brr$.

        Now, assume for an arbitrary $\bc \supseteq \bb$ and $L$ that $\Vdash_\bc^L \brl 1 \brr$. Further, for an arbitrary $\bd \supseteq \bc$ and $K$, assume that $\Vdash_\bd^K \bot$. By soundness of phase semantics $\brl 1 \brr \subseteq \brl \bot^\bot \brr$, hence $\brl 1 \brr \Vdash_\bb^\varnothing \brl \bot^\bot \brr$. By~\ref{eq:supp-inf}, $\Vdash_\bb^L \brl \bot^\bot \brr$. By inductive hypothesis, $\Vdash_\bb^L \bot^\bot$. By~\ref{eq:supp-inf} again, $\Vdash_\bd^{K,L} \bot$, and since $\bd \supseteq \bc$ was such that $\Vdash_\bd^K \bot$, obtain $\Vdash_\bc^L 1$ by~\ref{eq:supp-1}. Hence $\brl 1 \brr \Vdash_\bb^\varnothing 1$ by~\ref{eq:supp-inf}.

        \item[$\phi = \alpha \parr \beta:$] Suppose for an arbitrary $\bc \supseteq \bb$ and $L$ that $\Vdash_\bc^L \alpha \parr \beta$. Since $\brl \alpha \brr \otimes_M \brl \alpha^\bot \brr \subseteq \bot$ and $\brl \beta \brr \otimes_M \brl \beta^\bot \brr \subseteq \bot$, conclude $\brl \alpha \brr, \brl \alpha^\bot \brr \myv \bot$ and $\brl \beta \brr, \brl \beta^\bot \brr \myv \bot$, analogously written as $\brl \alpha \brr \Vdash_\bb^{\brl \alpha^\bot \brr} \bot$ and $\brl \beta \brr \Vdash_\bb^{\brl \beta^\bot \brr} \bot$. By inductive hypothesis then $\alpha \Vdash_\bb^{\brl \alpha^\bot \brr} \bot$ and $\beta \Vdash_\bb^{\brl \beta^\bot \brr} \bot$. Since $\Vdash_\bc^L \alpha \parr \beta$, by~\ref{eq:supp-parr}, obtain $\Vdash_\bc^{L,\brl \alpha^\bot \brr,\brl \beta^\bot \brr} \bot$,~i.e.~$\brl \alpha^\bot \brr,\brl \beta^\bot \brr \Vdash_\bc^{L} \bot$. Since $\brl (\alpha \parr \beta)^\bot \brr \subseteq \brl \alpha^\bot \brr \otimes_M \brl \beta^\bot \brr$, we have that $\brl (\alpha \parr \beta)^\bot \brr \myv \brl \alpha^\bot \brr \otimes \brl \beta^\bot \brr$. Hence, by~$(\otimes E)'$~\cite[Theorem 4.2]{DBLP:journals/corr/abs-2504-08349}, $\brl (\alpha \parr \beta)^\bot \brr \Vdash_\bc^{L} \bot$, and by \textit{reductio ad absurdum}, $\Vdash_\bc^{L} \brl \alpha \parr \beta \brr$. Thus, by~\ref{eq:supp-inf}, $\alpha \parr \beta \Vdash_\bb^\varnothing \brl \alpha \parr \beta \brr$.

        Now, assume for an arbitrary $\bc \supseteq \bb$ and $L$ that $\Vdash_\bc^L \brl \alpha \parr \beta \brr$. Further, for an arbitrary $\bd \supseteq \bc$ and atomic multisets $K,M$, assume that $\alpha \Vdash_\bd^K \bot$ and $\beta \Vdash_\bd^M \bot$. By inductive hypothesis, $\Vdash^K_\bd \brl \alpha^\bot \brr$ and $\Vdash^M_\bd \brl \beta^\bot \brr$. Since $\brl \alpha^\bot \brr \otimes_M \brl \beta^\bot \brr \subseteq \brl (\alpha \parr \beta)^\bot \brr$, we have that $\brl \alpha^\bot \brr, \brl \beta^\bot \brr \myv \brl (\alpha \parr \beta)^\bot \brr$. By~\ref{eq:supp-inf}, $\Vdash_\bd^{K,M} \brl (\alpha \parr \beta)^\bot \brr$. Since $\brl \alpha \parr \beta \brr \otimes_M \brl (\alpha \parr \beta)^\bot \brr \subseteq \bot$, hence $\brl \alpha \parr \beta \brr, \brl (\alpha \parr \beta)^\bot \brr \myv \bot$, obtain $\Vdash_\bd^{L,K,M} \bot$ by~\ref{eq:supp-inf}. Since $\bd \supseteq \bc$ was such that $\alpha \Vdash_\bd^K \bot$ and $\beta \Vdash_\bd^M \bot$, obtain $\Vdash_\bc^L \alpha \parr \beta$ by~\ref{eq:supp-parr}. Hence $\brl \alpha \parr \beta \brr \Vdash_\bb^\varnothing \alpha \parr \beta$ by~\ref{eq:supp-inf}.
        
        \item[$\phi = \alpha \with \beta:$] Suppose for an arbitrary $\bc \supseteq \bb$ and $L$ that $\Vdash_\bc^L \alpha \with \beta$. Then $\Vdash_\bc^L \alpha$ and $\Vdash_\bc^L \beta$ by~\ref{eq:supp-and}. Hence $\Vdash_\bc^L \brl \alpha \brr$ and $\Vdash_\bc^L \brl \beta \brr$ by inductive hypothesis, hence $L \vdash_\bc \brl \alpha \brr$ and $L \vdash_\bc \brl \beta \brr$ by Lemma~\ref{lmm:deriveqsupport}. Since $\brl \alpha \brr \cap \brl \beta \brr \subseteq \brl \alpha \with \beta \brr$, the following is a rule in $\bb$:
        \begin{prooftree}
        \def\ScoreOverhang{0.5pt}
            \AxiomC{$L \vdash \brl \alpha \brr $} 
            \AxiomC{$L \vdash \brl \beta \brr $} 
            \AxiomC{$\brl (\alpha \with \beta)^\bot \brr; \brl \alpha \with \beta \brr \vdash \bot$} 
            \TrinaryInfC{$L,\brl (\alpha \with \beta)^\bot \brr \vdash \bot$}
        \end{prooftree}
        Hence $L,\brl (\alpha \with \beta)^\bot \brr \vdash_\bc \bot$. By \textit{reductio ad absurdum} conclude $L \vdash_\bc \brl \alpha \with \beta \brr$, hence $\Vdash^L_\bc \brl \alpha \with \beta \brr$. By~\ref{eq:supp-inf}, $\alpha \with \beta \Vdash_\bb^\varnothing \brl \alpha \with \beta \brr$.

        Now, assume for an arbitrary $\bc \supseteq \bb$ and $L$ that $\Vdash_\bc^L \brl \alpha \with \beta \brr$. Since $\brl \alpha \with \beta \brr \subseteq \brl \alpha \brr$ and $\brl \alpha \with \beta \brr \subseteq \brl \beta \brr$, we can conclude $\brl \alpha \with \beta \brr \Vdash^\varnothing_\bb \brl \alpha \brr$ and $\brl \alpha \with \beta \brr \Vdash^\varnothing_\bb \brl \beta \brr$. Since $\Vdash_\bc^L \brl \alpha \with \beta \brr$, obtain $\Vdash^L_\bc \brl \alpha \brr$ and $\Vdash^L_\bc \brl \beta \brr$ by~\ref{eq:supp-inf}. Hence $\Vdash^L_\bc \alpha$ and $\Vdash^L_\bc \beta$ by inductive hypothesis, hence $\Vdash^L_\bc \alpha \with \beta$ by~\ref{eq:supp-and}. By~\ref{eq:supp-inf} then, $\brl \alpha \with \beta \brr \Vdash_\bb^\varnothing \alpha \with \beta$.
        
        \item[$\phi = \alpha \oplus \beta:$] Suppose for an arbitrary $\bc \supseteq \bb$ and $L$ that $\Vdash_\bc^L \alpha \oplus \beta$. Since $\brl \alpha \brr \subseteq \brl \alpha \oplus \beta \brr$ and $\brl \beta \brr \subseteq \brl \alpha \oplus \beta \brr$, conclude $\brl \alpha \brr \Vdash_\bb^\varnothing \brl \alpha \oplus \beta \brr$ and $\brl \beta \brr \Vdash_\bb^\varnothing \brl \alpha \oplus \beta \brr$. By inductive hypothesis then $\alpha \Vdash_\bb^\varnothing \brl \alpha \oplus \beta \brr$ and $\beta \Vdash_\bb^\varnothing \brl \alpha \oplus \beta \brr$. Since $\Vdash_\bc^L \alpha \oplus \beta$, by Lemma~\ref{lemma:genericplus}, obtain $\Vdash_\bc^L \brl \alpha \oplus \beta \brr$. Hence $\alpha \oplus \beta \Vdash_\bb^\varnothing \brl \alpha \oplus \beta \brr$ by~\ref{eq:supp-inf}.

        Now, assume for an arbitrary $\bc \supseteq \bb$ and $L$ that $\Vdash_\bc^L \brl \alpha \oplus \beta \brr$. Further, for an arbitrary $\bd \supseteq \bc$ and $K$, assume that $\alpha \Vdash_\bd^K \bot$ and $\beta \Vdash_\bd^K \bot$. By inductive hypothesis and Lemma~\ref{lmm:deriveqsupport}, $K \vdash_\bd \brl \alpha^\bot \brr$ and $K \vdash_\bd \brl \beta^\bot \brr$. Since $\brl \alpha^\bot \brr \cap \brl \beta^\bot \brr \subseteq \brl (\alpha \oplus \beta)^\bot \brr$, the following is a rule in $\bb$:
        \begin{prooftree}
        \def\ScoreOverhang{0.5pt}
            \AxiomC{$K \vdash \brl \alpha^\bot \brr $} 
            \AxiomC{$K \vdash \brl \beta^\bot \brr $} 
            \AxiomC{$\brl \alpha \oplus \beta \brr ; \brl (\alpha \oplus \beta)^\bot \brr \vdash \bot$} 
            \TrinaryInfC{$K,\brl \alpha \oplus \beta \brr \vdash \bot$}
        \end{prooftree}
        Hence $K,\brl \alpha \oplus \beta \brr \vdash_\bd \bot$. Since also $\Vdash_\bc^L \brl \alpha \oplus \beta \brr$, by~\ref{eq:supp-at}, obtain $K,L \vdash_\bd \bot$,~i.e.~$\Vdash^{K,L}_\bd \bot$. Since $\bd \supseteq \bc$ was such that $\alpha \Vdash_\bd^K \bot$ and $\beta \Vdash_\bd^K \bot$, obtain $\Vdash_\bc^L \alpha \oplus \beta$ by~\ref{eq:supp-plus}. By~\ref{eq:supp-inf} again, $\brl \alpha \oplus \beta \brr \Vdash_\bb^\varnothing \alpha \oplus \beta$.

        \item[$\phi = \top:$] Suppose for an arbitrary $\bc \supseteq \bb$ and $L$ that $\Vdash_\bc^L \top$. By soundness of phase semantics, $\brl K \brr \subseteq \brl \top \brr$ for all $\brl K \brr \in P$, then, in particular, $L \myv \brl \top \brr$,~i.e.~$\Vdash_\bb^L \brl \top \brr$. By~\ref{eq:supp-inf}, $\top \Vdash_\bb^\varnothing \brl \top \brr$.
        
        Now, assume for an arbitrary $\bc \supseteq \bb$ and $L$ that $\Vdash_\bc^L \brl \top \brr$. Since $\Vdash_\bb^K \top$ holds for all $\bc \supseteq \bb$ and $K$, in particular $\Vdash_\bc^L \top$. Hence $\brl \top \brr \Vdash_\bb^\varnothing \top$ by~\ref{eq:supp-inf}.

        \item[$\phi = 0:$] Suppose for an arbitrary $\bc \supseteq \bb$ and $L$ that $\Vdash_\bc^L 0$. By~\ref{eq:supp-0}, $\Vdash_\bc^{L,K} \bot$ for all $K$. Choose $K = \brl 0^\bot \brr$. Then $\Vdash_\bc^{L,\brl 0^\bot \brr} \bot$,~i.e.~$\brl 0^\bot \brr \Vdash_\bc^{L} \bot$, and by \textit{reductio ad absurdum}, $\Vdash_\bc^L \brl 0 \brr$. Hence $0 \Vdash_\bb^\varnothing \brl 0 \brr$ by~\ref{eq:supp-inf}.

        Now, assume for an arbitrary $\bc \supseteq \bb$ and $L$ that $\Vdash_\bc^L \brl 0 \brr$. By soundness of phase semantics, $\brl 0 \brr \otimes_M \brl K \brr \subseteq \bot$ for all $K$. Hence, $K, \brl 0 \brr, \myv \bot$,~i.e.~$\brl 0 \brr \Vdash_\bb^K \bot$. By~\ref{eq:supp-inf}, $\Vdash_\bc^{K,L} \bot$ for all $K$, hence $\Vdash_\bc^L 0$ by~\ref{eq:supp-0}. Hence $\brl 0 \brr \Vdash_\bb^\varnothing 0$ by~\ref{eq:supp-inf}.
        
        \item[$\phi = \bang \alpha:$] Suppose for an arbitrary $\bc \supseteq \bb$ and $L$ that $\Vdash_\bc^L \bang \alpha$. Further assume, for an arbitrary $\bd \supseteq \bc$ and an arbitrary structural $S$, that $\Vdash^S_\bd \alpha$. By inductive hypothesis, $\Vdash^S_\bd \brl \alpha \brr$. Since there is a rule in $\bb$ concluding $S \vdash \brl \bang \alpha \brr$  from $S \vdash \brl \alpha \brr$, we obtain $\Vdash^S_\bd \brl \bang \alpha \brr$. By soundness of phase semantics, $\brl \bang \alpha \brr \otimes_M \brl (\bang \alpha)^\bot \brr \subseteq \bot$, hence $\brl \bang \alpha \brr, \brl (\bang \alpha)^\bot \brr \myv \bot$. By~\ref{eq:supp-inf} then, $\brl (\bang \alpha)^\bot \brr \Vdash_\bd^S \bot$,~i.e.~$\Vdash_\bd^{S,\brl (\bang \alpha)^\bot \brr} \bot$. Since $\Vdash_\bc^L \bang \alpha$, conclude $\Vdash_\bc^{L,\brl (\bang \alpha)^\bot \brr} \bot$ by~\ref{eq:supp-bang}. Thus, $\brl (\bang \alpha)^\bot \brr \Vdash_\bc^{L} \bot$. By \textit{reductio ad absurdum}, $\Vdash_\bc^{L} \brl \bang \alpha \brr $, and by~\ref{eq:supp-inf}, $\bang \alpha \Vdash_\bb^\varnothing \brl \bang \alpha \brr$.

        Now, assume for an arbitrary $\bc \supseteq \bb$ and $L$ that $\Vdash_\bc^L \brl \bang \alpha \brr$. Further assume, for an arbitrary $\bd \supseteq \bc$, all $\be \supseteq \bd$, all structural $S$ and an arbitrary $K$, that $\Vdash_\be^S \alpha$ implies $\Vdash_\be^{S,K} \bot$. Since $\brl \bang \alpha \brr \subseteq \brl \alpha \brr$, we have that $\brl \bang \alpha \brr \Vdash_\bb^\varnothing \brl \alpha \brr$,~i.e.~$\Vdash_\bb^{\brl \bang \alpha \brr} \brl \alpha \brr$. By inductive hypothesis, $\Vdash_\bb^{\brl \bang \alpha \brr} \alpha$. Since $\brl \bang \alpha \brr$ is structural in $\bb$ (see Lemma~\ref{lmm:structatinb}), $\bd$ is its own extension and by monotonicity, obtain $\Vdash_\bd^{\brl \bang \alpha \brr, K} \bot$. This is equivalent to $\brl \bang \alpha \brr \Vdash_\bd^{K} \bot$, and since $\Vdash_\bc^L \brl \bang \alpha \brr$, obtain $\Vdash_\bd^{K,L} \bot$ by~\ref{eq:supp-inf}. Now, since $\bd \supseteq \bc$ was arbitrary and $\be \supseteq \bd$ was such that $\Vdash_\be^S \alpha$ implies $\Vdash_\be^{S,K} \bot$, obtain $\Vdash_\bc^L \bang \alpha$ by~\ref{eq:supp-bang}. By~\ref{eq:supp-inf} then, $\brl \bang \alpha \brr \Vdash_\bb^\varnothing \bang \alpha$.

        \item[$\phi = \quest \alpha:$] Suppose for an arbitrary $\bc \supseteq \bb$ and $L$ that $\Vdash_\bc^L \quest \alpha$. Notice that $\brl \bang(\alpha^\bot) \brr \subseteq \brl \alpha^\bot \brr$, hence $\brl \bang(\alpha^\bot) \brr \myv \brl \alpha^\bot \brr$, By inductive hypothesis, $\brl \bang(\alpha^\bot) \brr \myv \alpha^\bot$,~i.e.~$\Vdash_\bb^{\brl \bang(\alpha^\bot) \brr} \alpha^\bot$. Since $\brl \bang (\alpha^\bot) \brr$ is structural in $\bb$ (see Lemma~\ref{lmm:structatinb}), obtain $\Vdash_\bc ^{L,\brl \bang(\alpha^\bot) \brr} \bot$,~i.e.~$\brl \bang(\alpha^\bot) \brr \Vdash_\bc ^{L} \bot$ by~\ref{eq:supp-quest}. By soundness of phase semantics, $\brl (\quest \alpha)^\bot \brr \subseteq \brl \bang(\alpha^\bot) \brr$, hence $\brl (\quest \alpha)^\bot \brr \myv \brl \bang(\alpha^\bot) \brr$. By~\ref{eq:supp-inf}, $\brl (\quest \alpha)^\bot \brr \Vdash_\bc ^{L} \bot$, and by \textit{reductio ad absurdum}, $\Vdash_\bc ^{L} \brl \quest \alpha \brr$. By~\ref{eq:supp-inf} then, $\quest \alpha \Vdash_\bb^\varnothing \brl \quest \alpha \brr$.
        
        Now, assume for an arbitrary $\bc \supseteq \bb$ and $L$ that $\Vdash_\bc^L \brl \quest \alpha \brr$. Further assume, for an arbitrary $\bd \supseteq \bc$ and structural $S$, that $\Vdash_\bd^S \alpha^\bot$. By inductive hypothesis, $\Vdash_\bd^S \brl \alpha^\bot \brr$. Since there is a rule in $\bb$ concluding $S \vdash \brl \bang (\alpha^\bot) \brr$  from $S \vdash \brl \alpha^\bot \brr$, we obtain $\Vdash_\bd^S \brl \bang (\alpha^\bot) \brr$. By soundness of phase semantics, $\brl \bang(\alpha^\bot) \brr \subseteq \brl (\quest \alpha)^\bot \brr$, hence $\brl \bang(\alpha^\bot) \brr \myv \brl (\quest \alpha)^\bot \brr$. By~\ref{eq:supp-inf}, $\Vdash_\bd^S \brl (\quest \alpha)^\bot \brr$, and since $\brl \quest \alpha \brr, \brl (\quest \alpha)^\bot \brr \myv \bot$, by \ref{eq:supp-inf} again, $\Vdash_\bd^{S,L} \bot$. Since $\bd \supseteq \bc$ was such that $\Vdash_\bd^S \alpha^\bot$, conclude $\Vdash_\bc^L \quest \alpha$ by~\ref{eq:supp-quest}. By~\ref{eq:supp-inf} then, $\brl \quest \alpha \brr \Vdash_\bb^\varnothing \quest \alpha$.
    \end{description}
\end{proof}

\begin{proof}[Proof of Lemma~\ref{lmm:validphasespace}]
    We prove the statement by induction on the structure of $\phi$.
    \begin{description}
        \item[$\phi = p:$] $p^* = \{\Gamma \; | \; \Gamma \Vdash_\bb^\varnothing p \}$ by definition.
        
        \item[$\phi = \alpha^\bot:$] Suppose $\Gamma \in (\alpha^\bot)^* = ((\alpha)^*)^\bot = \{\Gamma \; | \; \Gamma \Vdash_\bb^\varnothing \alpha \}^\bot$ by inductive hypothesis. Then, $\forall \Delta \in \{\Gamma \; | \; \Gamma \Vdash_\bb^\varnothing \alpha \}, \Gamma,\Delta \Vdash_\bb^\varnothing \bot$. Since $\alpha \Vdash^\varnothing_\bb \alpha$, choose $\Delta = \{\alpha\}$,~i.e.~$\Gamma,\alpha \Vdash_\bb^\varnothing \bot$, which is equivalent to $\Gamma \Vdash_\bb^\varnothing \alpha^\bot$. Hence, $\Gamma \in \{\Gamma \; | \; \Gamma \Vdash_\bb^\varnothing \alpha^\bot \}$.
        
        Conversely, suppose that $\Gamma \in \{\Gamma \; | \; \Gamma \Vdash_\bb^\varnothing \alpha^\bot \}$ and assume that $\Delta \in \{\Gamma \; | \; \Gamma \Vdash_\bb^\varnothing \alpha \}$. Then $\Gamma \myv \alpha^\bot$ and $\Delta \myv \alpha$. Hence, $\Gamma, \alpha \myv \bot$ and by~\ref{eq:supp-inf}, $\Gamma, \Delta \myv \bot$. Thus, $\Gamma \in \{\Gamma \; | \; \Gamma \Vdash_\bb^\varnothing \alpha \}^\bot = ((\alpha)^*)^\bot$ by inductive hypothesis. Hence, $\Gamma \in (\alpha^\bot)^*$.
        
        \item[$\phi = \alpha \otimes \beta:$] Suppose $\Gamma \in (\alpha \otimes \beta)^* = (\alpha)^* \otimes_M (\beta)^* = \{\Gamma \; | \; \Gamma \Vdash_\bb^\varnothing \alpha \} \otimes_M  \{\Gamma \; | \; \Gamma \Vdash_\bb^\varnothing \beta \} = (\{\Gamma \; | \; \Gamma \Vdash_\bb^\varnothing \alpha \} \uplus \{\Gamma \; | \; \Gamma \Vdash_\bb^\varnothing \beta \})^{\bot\bot}$ by inductive hypothesis. Hence, $\Gamma, \Delta \myv \bot$ for all $\Delta$ s.t. $\Sigma,\Theta,\Delta \myv \bot$ whenever $\Sigma \myv \alpha$ and $\Theta \myv \beta$. By~$(\otimes I)'$~\cite[Theorem 4.2]{DBLP:journals/corr/abs-2504-08349}, $\alpha,\beta \myv \alpha \otimes \beta$, hence $\alpha, \beta, (\alpha \otimes \beta)^\bot \myv \bot$, thus $\Sigma, \Theta, (\alpha \otimes \beta)^\bot \myv \bot$. Now, choose $\Delta = \{(\alpha \otimes \beta)^\bot\}$, so that we can conclude $\Gamma, (\alpha \otimes \beta)^\bot \myv \bot$. Hence, $\Gamma \myv \alpha \otimes \beta$ by Lemma~\ref{lmm:raa}, hence $\Gamma \in \{\Gamma \; | \; \Gamma \Vdash_\bb^\varnothing \alpha \otimes \beta \}$.
        
        Conversely, suppose that $\Gamma \in \{\Gamma \; | \; \Gamma \Vdash_\bb^\varnothing \alpha \otimes \beta \}$ and assume that $\Delta \in (\{\Gamma \; | \; \Gamma \Vdash_\bb^\varnothing \alpha \} \uplus \{\Gamma \; | \; \Gamma \Vdash_\bb^\varnothing \beta \})^\bot$,~i.e.~that $\Delta, \Sigma,\Theta \myv \bot$ whenever $\Sigma \myv \alpha$ and $\Theta \myv \beta$. Choose $\Sigma = \{\alpha\}$ and $\Theta = \beta$, hence $\alpha, \beta,\Delta \myv \bot$. Assume, for an arbitrary $\bc \supseteq \bb$ and atomic multisets $L,K$, that $\Vdash_\bc^L \Gamma$ and $\Vdash_\bc^K \Delta$. Hence, $\Vdash_\bc^L \alpha \otimes \beta$ and $\alpha,\beta \Vdash_\bc^K \bot$. By~\ref{eq:supp-tensor}, obtain $\Vdash_\bc^{L,K} \bot$, hence, by~\ref{eq:supp-inf}, $\Gamma,\Delta \myv \bot$. Thus, $\Gamma \in (\{\Gamma \; | \; \Gamma \Vdash_\bb^\varnothing \alpha \} \uplus \{\Gamma \; | \; \Gamma \Vdash_\bb^\varnothing \beta \})^{\bot\bot} = \{\Gamma \; | \; \Gamma \Vdash_\bb^\varnothing \alpha \} \otimes_M  \{\Gamma \; | \; \Gamma \Vdash_\bb^\varnothing \beta \} = (\alpha)^* \otimes_M (\beta)^*$ by inductive hypothesis. Hence, $\Gamma \in (\alpha \otimes \beta)^*$.

        \item[$\phi = \alpha \parr \beta:$] Suppose $\Gamma \in (\alpha \parr \beta)^* = (\alpha)^* \parr_M (\beta)^* = \{\Gamma \; | \; \Gamma \Vdash_\bb^\varnothing \alpha \} \parr_M  \{\Gamma \; | \; \Gamma \Vdash_\bb^\varnothing \beta \} = (\{\Gamma \; | \; \Gamma \Vdash_\bb^\varnothing \alpha \}^\bot \uplus \{\Gamma \; | \; \Gamma \Vdash_\bb^\varnothing \beta \}^\bot)^{\bot} = (\{\Gamma \; | \; \Gamma \Vdash_\bb^\varnothing \alpha^\bot\} \uplus \{\Gamma \; | \; \Gamma \Vdash_\bb^\varnothing \beta^\bot \} \brr)^{\bot}$ by inductive hypothesis. Hence, $\Gamma,\Delta,\Sigma \myv \bot$ whenever $\Delta \myv \alpha^\bot$ and $\Sigma \myv \beta^\bot$. Choose $\Delta = \{\alpha^\bot\}$ and $\Sigma = \{\beta^\bot\}$. Hence, $\Gamma, \alpha^\bot, \beta^\bot \myv \bot$. By~$(\parr I)'$~\cite[Theorem 4.2]{DBLP:journals/corr/abs-2504-08349}, $\Gamma \myv \alpha \parr \beta$, hence $\Gamma \in \{\Gamma \; | \; \Gamma \Vdash_\bb^\varnothing \alpha \parr \beta\}$.

        Conversely, suppose that $\Gamma \in \{\Gamma \; | \; \Gamma \Vdash_\bb^\varnothing \alpha \parr \beta \}$ and assume that $\Delta \in \{\Gamma \; | \; \Gamma \Vdash_\bb^\varnothing \alpha^\bot \}$ and $\Sigma \in \{\Gamma \; | \; \Gamma \Vdash_\bb^\varnothing \beta^\bot \}$,~i.e.~that $\Delta \myv \alpha^\bot$ and $\Sigma \myv \beta^\bot$. Now, assume, for an arbitrary $\bc \supseteq \bb$ and atomic multisets $L,K,M$, that $\Vdash_\bc^L \Gamma$, $\Vdash_\bc^K \Delta$ and $\Vdash_\bc^M \Sigma$. Hence, $\Vdash_\bc^L \alpha \parr \beta$ and $\Vdash_\bc^K \alpha^\bot$ and $\Vdash_\bc^M \beta^\bot$. By~\ref{eq:supp-parr}, conclude $\Vdash_\bc^{L,K,M} \bot$. Hence, by~\ref{eq:supp-inf}, $\Gamma,\Delta,\Sigma \myv \bot$, thus $\Gamma \in (\{\Gamma \; | \; \Gamma \Vdash_\bb^\varnothing \alpha^\bot\} \uplus \{\Gamma \; | \; \Gamma \Vdash_\bb^\varnothing \beta^\bot \}r)^{\bot}$. Hence, $\Gamma \in (\{\Gamma \; | \; \Gamma \Vdash_\bb^\varnothing \alpha \}^\bot \uplus \{\Gamma \; | \; \Gamma \Vdash_\bb^\varnothing \beta \}^\bot)^{\bot} = \{\Gamma \; | \; \Gamma \Vdash_\bb^\varnothing \alpha \} \parr_M  \{\Gamma \; | \; \Gamma \Vdash_\bb^\varnothing \beta \} = (\alpha)^* \parr_M (\beta)^*$ by inductive hypothesis,~i.e.~$\Gamma \in (\alpha \parr \beta)^*$.

        \item[$\phi = 1:$] Suppose $\Gamma \in 1^* = \bot^\bot$. Then, $\forall\Delta \in \bot$, $\Gamma,\Delta \myv \bot$. Choose $\Delta = \{\bot\}$, hence $\Gamma,\bot \myv \bot$. By~$(\multimap I)'$~\cite[Theorem 4.2]{DBLP:journals/corr/abs-2504-08349} and~\ref{eq:supp-1}, $\Gamma \myv 1$, hence $\Gamma \in \{\Gamma \; | \; \Gamma \Vdash_\bb^\varnothing 1 \}$.

        Conversely, suppose that $\Gamma \in \{\Gamma \; | \; \Gamma \Vdash_\bb^\varnothing 1 \}$,~i.e.~that $\Gamma \myv 1$, and assume that $\Delta \in \bot$. Now, assume, for an arbitrary $\bc \supseteq \bb$ and $L,K$, that $\Vdash_\bc^L \Gamma$ and $\Vdash_\bc^K \Delta$. Hence, $\Vdash_\bc^L 1$ and $\Vdash_\bc^K \bot$. By~\ref{eq:supp-1}, conclude $\Vdash_\bc^{L,K} \bot$, hence by~\ref{eq:supp-inf}, $\Gamma,\Delta\myv \bot$. Thus, $\Gamma \in \bot^\bot = 1^*$.

        \item[$\phi = \alpha \with \beta:$] Suppose $\Gamma \in (\alpha \with \beta)^* = (\alpha)^* \with_M (\beta)^* = \{\Gamma \; | \; \Gamma \Vdash_\bb^\varnothing \alpha \} \with_M \{\Gamma \; | \; \Gamma \Vdash_\bb^\varnothing \beta \} = \{\Gamma \; | \; \Gamma \Vdash_\bb^\varnothing \alpha \} \cap \{\Gamma \; | \; \Gamma \Vdash_\bb^\varnothing \beta \}$ by inductive hypothesis. That means that $\Gamma \in \{\Gamma \; | \; \Gamma \Vdash_\bb^\varnothing \alpha \}$ and $\Gamma \in \{\Gamma \; | \; \Gamma \Vdash_\bb^\varnothing \beta \}$,~i.e.~$\Gamma \myv \alpha$ and $\Gamma \myv \beta$. By~\ref{eq:supp-and}, $\Gamma \myv \alpha \with \beta$, hence $\Gamma \in \{\Gamma \; | \; \Gamma \Vdash_\bb^\varnothing \alpha \with \beta\}$.

        Conversely, suppose that $\Gamma \in \{\Gamma \; | \; \Gamma \Vdash_\bb^\varnothing \alpha \with \beta \}$,~i.e.~that $\Gamma \myv \alpha \with \beta$. By~\ref{eq:supp-and}, $\Gamma \myv \alpha$ and $\Gamma \myv \beta$, hence $\Gamma \in \{\Gamma \; | \; \Gamma \Vdash_\bb^\varnothing \alpha \}$ and $\Gamma \in \{\Gamma \; | \; \Gamma \Vdash_\bb^\varnothing \beta \}$, so $\Gamma \in \{\Gamma \; | \; \Gamma \Vdash_\bb^\varnothing \alpha \} \cap \{\Gamma \; | \; \Gamma \Vdash_\bb^\varnothing \beta \} = \{\Gamma \; | \; \Gamma \Vdash_\bb^\varnothing \alpha \} \with_M \{\Gamma \; | \; \Gamma \Vdash_\bb^\varnothing \beta \} = (\alpha)^* \with_M (\beta)^* $ by inductive hypothesis. Hence, $\Gamma \in (\alpha \with \beta)^*$.

        \item[$\phi = \alpha \oplus \beta:$] Suppose $\Gamma \in (\alpha \oplus \beta)^* = (\alpha)^* \oplus_M (\beta)^* =\{\Gamma \; | \; \Gamma \Vdash_\bb^\varnothing \alpha \} \oplus_M \{\Gamma \; | \; \Gamma \Vdash_\bb^\varnothing \beta \} = (\{\Gamma \; | \; \Gamma \Vdash_\bb^\varnothing \alpha \} \cup \{\Gamma \; | \; \Gamma \Vdash_\bb^\varnothing \beta \})^{\bot\bot}$ by inductive hypothesis. Hence, $\Gamma, \Delta \myv \bot$ for all $\Delta$ s.t. $\Sigma,\Delta \myv \bot$ whenever $\Sigma \myv \alpha$ or $\Sigma \myv \beta$. We know that $\alpha \myv \alpha \oplus \beta$ and $\beta \myv \alpha \oplus \beta$ by~$(\oplus I)'$~\cite[Theorem 4.2]{DBLP:journals/corr/abs-2504-08349}. Hence, $\Sigma \myv \alpha \oplus \beta$, thus $\Sigma, (\alpha \oplus \beta)^\bot \myv \bot$. Choose $\Delta = \{(\alpha \oplus \beta)^\bot\}$, so that we can conclude $\Gamma, (\alpha \oplus \beta)^\bot \myv \bot$. By Lemma~\ref{lmm:raa}, $\Gamma \myv \alpha \oplus \beta$, hence $\Gamma \in \{\Gamma \; | \; \Gamma \Vdash_\bb^\varnothing \alpha \oplus \beta \}$.

        Conversely, suppose that $\Gamma \in \{\Gamma \; | \; \Gamma \Vdash_\bb^\varnothing \alpha \oplus \beta\}$,~i.e.~$\Gamma \myv \alpha \oplus \beta$, and assume that $\Delta \in (\{\Gamma \; | \; \Gamma \Vdash_\bb^\varnothing \alpha \} \cup \{\Gamma \; | \; \Gamma \Vdash_\bb^\varnothing \beta \})^\bot$,~i.e.~that $\Delta, \Sigma \myv \bot$ whenever $\Sigma \myv \alpha$ or $\Sigma \myv \beta$. For $\Sigma = \alpha$ obtain $\Delta, \alpha \myv \bot$, and for $\Sigma = \beta$ obtain $\Delta, \beta \myv \bot$. Assume, for an arbitrary $\bc \supseteq \bb$ and atomic multisets $L,K$, that $\Vdash_\bc^L \Gamma$ and $\Vdash_\bc^K \Delta$. Hence, $\Vdash_\bc^L \alpha \oplus \beta$ and $\alpha \Vdash_\bc^K \bot$ and $\beta \Vdash_\bc^K \bot$. By~\ref{eq:supp-plus}, obtain $\Vdash_\bc^{L,K} \bot$, hence, by~\ref{eq:supp-inf}, $\Gamma,\Delta \myv \bot$. Thus, $\Gamma \in (\{\Gamma \; | \; \Gamma \Vdash_\bb^\varnothing \alpha \} \cup \{\Gamma \; | \; \Gamma \Vdash_\bb^\varnothing \beta \})^{\bot\bot} = \{\Gamma \; | \; \Gamma \Vdash_\bb^\varnothing \alpha \} \oplus_M  \{\Gamma \; | \; \Gamma \Vdash_\bb^\varnothing \beta \} = (\alpha)^* \oplus_M (\beta)^*$ by inductive hypothesis. Hence, $\Gamma \in (\alpha \oplus \beta)^*$.
        
        \item[$\phi = \top:$] Suppose $\Gamma \in \top^* = M$,~i.e.~$\Gamma$ is any multiset. But by~\ref{eq:supp-top}, $\Gamma \myv \top$ always holds, hence $\Gamma \in \{\Gamma \; | \; \Gamma \Vdash_\bb^\varnothing \top \}$.

        Conversely, suppose that $\Gamma \in \{\Gamma \; | \; \Gamma \Vdash_\bb^\varnothing \top \}$,~i.e.~that $\Gamma \myv \top$, which always holds by~\ref{eq:supp-top}, hence $\Gamma$ can be any multiset,~i.e.~$\Gamma \in M = \top^*$.

        \item[$\phi = 0:$] Suppose $\Gamma \in 0^* = (\top^*)^\bot = \{\Gamma \; | \; \Gamma \Vdash_\bb^\varnothing \top \}^\bot$ by inductive hypothesis. Hence, $\Gamma, \Delta \myv \bot$ for all $\Delta$ s.t. $\Delta \myv \top$,~i.e.~for all $\Delta \in M$. Choose $\Delta = \{\top\}$, hence $\Gamma, \top \myv \bot$, and by~$(\multimap I)'$~\cite[Theorem 4.2]{DBLP:journals/corr/abs-2504-08349},~\ref{eq:supp-top} and~\ref{eq:supp-0}, obtain $\Gamma \myv 0$, hence $\Gamma \in \{\Gamma \; | \; \Gamma \Vdash_\bb^\varnothing 0 \}$.

        Conversely, suppose that $\Gamma \in \{\Gamma \; | \; \Gamma \Vdash_\bb^\varnothing 0 \}$,~i.e.~that $\Gamma \myv 0$ and assume that $\Delta \in \{\Gamma \; | \; \Gamma \Vdash_\bb^\varnothing \top \}$. Further assume, for an arbitrary $\bc \supseteq \bb$ and $L,K$ that $\Vdash_\bc^L \Gamma$ and $\Vdash_\bc^K \Delta$. Then $\Vdash_\bc^L 0$, hence $\Vdash_\bc^{L,M} \bot$ for all $M$ by~\ref{eq:supp-0}. In particular, $\Vdash_\bc^{L,K} \bot$. By~\ref{eq:supp-inf}, $\Gamma,\Delta \myv \bot$, hence $\Gamma \in \{\Gamma \; | \; \Gamma \Vdash_\bb^\varnothing \top \}^\bot = (\top^*)^\bot$ by inductive hypothesis. Hence, $\Gamma \in 0^*$.

        \item[$\phi = \bang \alpha:$] Suppose $\Gamma \in (\bang \alpha)^* = \bang_M (\alpha)^* = \bang_M \{\Gamma \; | \; \Gamma \Vdash_\bb^\varnothing \alpha \} = (\{\Gamma \; | \; \Gamma \Vdash_\bb^\varnothing \alpha \} \cap I)^{\bot\bot}$ by inductive hypothesis. Hence, $\Gamma, \Delta \myv \bot$ for all $\Delta$ s.t. $\Sigma,\Delta \myv \bot$ whenever $\Sigma \myv \alpha$ and $\Sigma \in I$. Notice that, if $\Sigma \in I$, then every $\phi \in \Sigma$ is of the shape $\bang \psi$ for some $\psi$. Hence, if $\Sigma \myv \alpha$, then $\Sigma \myv \bang \alpha$ by~$(\bang I)'$ (Theorem~\ref{thm:soundness}). Further notice that $\bang \alpha, (\bang \alpha)^\bot \myv \bot$, thus $\Sigma, (\bang \alpha)^\bot \myv \bot$. Now, choose $\Delta = \{(\bang \alpha)^\bot\}$, so that we can conclude $\Gamma,(\bang \alpha)^\bot \myv \bot$. By Lemma~\ref{lmm:raa}, $\Gamma \myv \bang \alpha$, hence $\Gamma \in \{\Gamma \; | \; \Gamma \Vdash_\bb^\varnothing \bang \alpha \}$.

        Conversely, suppose that $\Gamma \in \{\Gamma \; | \; \Gamma \Vdash_\bb^\varnothing \bang \alpha \}$,~i.e.~$\Gamma \myv \bang \alpha$, and assume that $\Delta \in (\{\Gamma \; | \; \Gamma \Vdash_\bb^\varnothing \alpha \} \cap I)^\bot$,~i.e.~that $\Delta, \Sigma \myv \bot$ whenever $\Sigma \myv \alpha$ and $\Sigma \in I$. Now, choose $\Sigma = \{\bang \alpha\}$. Then $\Delta, \bang \alpha \myv \bot$, and since $\Gamma \myv \bang \alpha$, obtain $\Delta,\Gamma \myv \bot$. Hence, $\Gamma \in (\{\Gamma \; | \; \Gamma \Vdash_\bb^\varnothing \alpha \} \cap I)^{\bot\bot} = \bang_M \{\Gamma \; | \; \Gamma \Vdash_\bb^\varnothing \alpha \}r = \bang_M (\alpha)^*$ by inductive hypothesis. Thus, $\Gamma \in (\bang \alpha)^*$.

        \item[$\phi = \quest \alpha:$] Suppose $\Gamma \in (\quest \alpha)^* = \quest_M (\alpha)^* = \quest_M \{\Gamma \; | \; \Gamma \Vdash_\bb^\varnothing \alpha \} = (\{\Gamma \; | \; \Gamma \Vdash_\bb^\varnothing \alpha \}^\bot \cap I)^{\bot} = (\{\Gamma \; | \; \Gamma \Vdash_\bb^\varnothing \alpha^\bot \} \cap I)^{\bot}$ by inductive hypothesis. Hence, $\Gamma,\Delta \myv \bot$ whenever $\Delta \myv \alpha^\bot$ and $\Delta \in I$. Choose $\Delta = \{\bang (\alpha^\bot)\}$, then $\Gamma, \bang (\alpha^\bot) \myv \bot$. By soundness of B-eS, $\Gamma \myv \quest \alpha$, hence $\Gamma \in \{\Gamma \; | \; \Gamma \Vdash_\bb^\varnothing \quest \alpha \}$.

        Conversely, suppose that $\Gamma \in \{\Gamma \; | \; \Gamma \Vdash_\bb^\varnothing \quest \alpha \}$,~i.e.~$\Gamma \myv \quest \alpha$, and assume that $\Delta \in \{\Gamma \; | \; \Gamma \Vdash_\bb^\varnothing \alpha^\bot \}$,~i.e.~$\Delta \myv \alpha^\bot$, and that $\Delta \in I$. Notice that, if $\Delta \in I$, then every $\phi \in \Delta$ is of the shape $\bang \psi$ for some $\psi$. Hence, if $\Delta \myv \alpha^\bot$, then $\Delta \myv \bang (\alpha^\bot)$ by~$(\bang I)'$ (Theorem~\ref{thm:soundness}). By soundness of B-eS, $\bang (\alpha^\bot) \myv \quest \alpha^\bot$, hence $\Delta \myv \quest \alpha^\bot$. Since also $\Gamma \myv \quest \alpha$, obtain $\Gamma,\Delta \myv \bot$. Hence $\Gamma \in (\{\Gamma \; | \; \Gamma \Vdash_\bb^\varnothing \alpha^\bot \} \cap I)^{\bot} = (\{\Gamma \; | \; \Gamma \Vdash_\bb^\varnothing \alpha \}^\bot \cap I)^{\bot} = \quest_M \{\Gamma \; | \; \Gamma \Vdash_\bb^\varnothing \alpha \} = \quest_M (\alpha)^*$ by inductive hypothesis. Thus, $\Gamma \in (\quest \alpha)^*$.
    \end{description}
\end{proof}

\end{document}